\documentclass[final,5p,times,twocolumn]{elsarticle}

\usepackage{hyperref}

\journal{Control Engineering Practice}

\usepackage{graphicx}      

\usepackage[cmex10]{amsmath}
\usepackage{amsfonts}
\usepackage{color}
\usepackage{psfrag}
\usepackage{subfigure}
\usepackage{empheq}
\usepackage{tabularx}
\usepackage{bm} 
\usepackage{nicefrac} 

\usepackage{url}
\usepackage{psfrag}
\usepackage{subfigure}

\usepackage{empheq}
\usepackage{tabularx}
\usepackage{multirow}
\usepackage{algorithm}
\usepackage{algpseudocode}
\usepackage{booktabs}

\usepackage{enumerate}

\usepackage{cleveref}

\usepackage[EULERGREEK]{sansmath}

\usepackage{pgfplots}
\pgfplotsset{compat=newest}

\usepackage{tikz}
\usepackage{colortbl}
\usepackage{url}
\usepackage{zref-savepos}

\usepackage{makecell}

\usetikzlibrary{external}
\tikzset{external/force remake}

\newtheorem{defn}{Definition} 
\newtheorem{thm}{Theorem} 

\newtheorem{lem}{Lemma} 
\newenvironment{proof}{\textit{Proof.}}{~\hfill\rule{0.75em}{0.75em}\\}

\newcommand{\R}{\mathbb{R}}
\newcommand{\Comp}{\mathbb{C}}

\newcommand{\cH}{\mathcal{H}}

\newcommand{\cP}{\mathcal{P}}
\newcommand{\cR}{\mathcal{R}}

\newcommand{\cS}{\mathcal{S}}

\newcommand{\vx}{\mathbf{x}}
\newcommand{\vw}{\mathbf{w}}
\newcommand{\vK}{\mathbf{K}}
\newcommand{\vy}{\mathbf{y}}

\newcommand{\vz}{\mathbf{z}}

\newcommand{\vP}{\mathbf{P}}

\newcommand{\vtheta}{\bm{\theta}}

\newcommand{\vDK}{\bm{\Delta K}}
\newcommand{\vQ}{\mathbf{Q}}

\newcommand{\myRe}{\text{Re}}

\newtheorem{rem}{Remark} 

\newcommand{\Hinf}{\cH_{\infty}}
\newcommand{\Htwo}{\cH_2}
\newcommand{\DefinedAs}[0]{\mathrel{\mathop:}=}

\newcommand{\bigSigma}{\overline{\sigma}}
\newcommand{\HinfNormG}{\left\| G(\vK_t, s) \right\|_\infty}

\newlength\fheight
\newlength\fwidth
\newlength\fheightTwo
\hypersetup{
	bookmarks=true,         
	unicode=false,          
	pdftoolbar=true,        
	pdfmenubar=true,        
	pdffitwindow=false,     
	pdftitle={Guaranteed H-infinity Controller Parameter Tuning for Power Systems: Method and Experimental Evaluation},    
	pdfauthor={Amer Mesanovic},     
	pdfcreator={Amer Mesanovic},   
	pdfproducer={word},  
	colorlinks=false,       
}

\newif\ifcommenttorolf
\commenttorolftrue

\begin{document}

\begin{frontmatter}

\title{Structured Controller Parameter Tuning for Power Systems\tnoteref{mytitlenote}}
\tnotetext[mytitlenote]{This work has been partially funded by the German Federal Ministry for Economic Affairs and Energy under Grant number 03ET7541A in the frame of the DynaGridCenter project, and under Grant number 0325685A in the frame of the IREN2 project.}


\author[mymainaddress,siemensmunich]{Amer Me{\v s}anovi{\'c}} 
\ead{amer.mesanovic@siemens.com}

\author[ulrichsaddress]{Ulrich M{\"u}nz}
\ead{ulrich.muenz@siemens.com}

\author[siemensmunich]{Andrei Szabo}
\author[siemenserl]{Martin Mangold}

\author[siemenserl]{Joachim Bamberger}
\author[siemensmunich]{Michael Metzger}
\author[siemenserl]{Chris Heyde}
\author[siemenserl]{Rainer Krebs}

\author[mymainaddress]{Rolf Findeisen}
\ead{rolf.findeisen@ovgu.de}

\address[mymainaddress]{Laboratory for Systems Theory and Automatic Control, Otto-von-Guericke-University Magdeburg, Germany}
\address[siemensmunich]{Siemens AG, Otto-Hahn-Ring 6, Munich, Germany}
\address[siemenserl]{Siemens AG, Freyeslebenstr. 1, Erlangen, Germany}
\address[ulrichsaddress]{Siemens Corp., 755 College Road East, Princeton, NJ, USA}

\begin{abstract}
	Reliable and secure operation of power systems becomes increasingly challenging as the share of volatile generation rises, leading to largely changing dynamics.
	Typically, the architecture and structure of controllers in power systems, such as voltage controllers of power generators, are fixed during the design and buildup of the network. 
	As replacing existing controllers is often undesired and challenging, 
	setpoint adjustments, as well as tuning of the controller parameters, are possibilities to counteract changing dynamics. 
	We present an approach for fast and computationally efficient adaptation of parameters of structured controllers based on $\Hinf$ optimization, also referred to as structured $\Hinf$ controller synthesis, tailored towards power systems. 
	{The goal of the tuning is to increase the robustness of the system towards disturbances.} 
	Conditions are established that guarantee that the approach leads to stability.
	The results are verified in a testbed microgrid consisting of six inverters and a load bank, as well as in several simulation studies. {Furthermore, the performance of the approach is compared to other tuning approaches, thereby demonstrating significantly reduced computation times.} 
	The proposed method improves the system robustness, as well as the time-response to step disturbances and allows structured controller tuning even for large networks.
\end{abstract}


\end{frontmatter}


\allowdisplaybreaks

\section{Introduction}

Reliable and secure electric power supply is vital for modern life. Power systems must operate without interruptions, despite unknown disturbances, such as outages of components, unknown load dynamics, and changes in power generation.
Power systems consist of prosumers, such as power plants, wind turbines and users, which are interconnected by a power grid, c.f. Fig.~\ref{fig.GridModel}.
The reliable and safe operation of power systems is ``guaranteed'' today by a complete automation system, consisting of, e.g., PID controllers, notch filters, and lead-lag filters, controlling power system components spanning from power plants to inverters, flexible AC transmission system elements and loads~\cite{kundur93a}. 
These automation systems result from careful considerations based on years of practical experience and operation.
Tuning of the corresponding controller parameters is very important for reliable operation.
{This is currently guaranteed by tuning and verification during the installation of a component.}
The resulting controllers are typically not re-parameterized until a large problem in the system occurs{~\cite{ENTSOEReport}}. Such manual tuning has proven to be sufficient as long as the network and power plants do not change significantly. While variations in the grid are constantly present due to load fluctuations or generator outages, these variations are often predictable and can be considered during the manual tuning procedure.


{Increasing amounts of renewable generation lead to large changes in the operation and the resulting dynamic behavior of the power systems~\cite{Ren21Report2018}.
Depending on the weather conditions, renewable generation can change constantly and can shift geographically. Furthermore, if the weather conditions are not suitable for renewable generation, the percentage of conventional generation needs to increase. These changes require new approaches e.g. for optimal power flow calculation or unit commitment.
In this work, we focus on the effects of renewable generation on power system dynamics.
The constantly shifting mix of renewable and conventional generation can lead to time-varying dynamics, i.e. oscillatory modes~\cite{AlAli14,crivellaro2019beyond, markovic2019understanding,mesanovic2019hierarchical}.}
If not handled, the controllers in large power systems, which are tuned today for fixed oscillatory modes, become less effective, increasing the risk of blackouts{~\cite{crivellaro2019beyond}. A simulation study showing the described effects can be found in~\cite{mesanovic2019hierarchical}}.
Thus, new control methods are necessary to improve the robustness of power networks and to account for the changing dynamics. {Furthermore, these methods will have to consider the dynamic behavior of conventional, as well as of inverter-based generation.}

We propose to adapt the parameters of controllers already present in power systems to the seemingly changing operating conditions. {This approach has been recognized in the literature as a possible solution to handle time-varying dynamics~\cite{crivellaro2019beyond, markovic2019understanding}. The overall tuning goal of the proposed approach is the reduction of the system $\Hinf$ norm, which increases the system robustness towards disturbances and typically improves the time-domain behavior. For this purpose, we introduce a controller parameter tuning method, which is applicable to a broad class of control structures and is scalable towards large scale systems, i.e. is applicable to power systems ranging from microgrids to transmission systems.}
We present an iterative convex optimization approach for structured $\Hinf$ controller synthesis of linear systems, which optimizes the parameters of existing controllers to current conditions in the system. We provide certificates which guarantee stability of the optimal tuned system and evaluate the effectiveness of the approach in realistic simulations, experimentally considering an inverter-only microgrid testbed.

Controller synthesis for power systems typically exploits $\Hinf$ optimization, $\Htwo$ optimization, and pole placement, c.f.~\cite{Raoufat16,Pipelzadeh17,Zhu03,befekadu2005robust,MahmoudiNAPS15,Preece13,wu2016input,Schuler14}.
Other control design and analysis approaches are sensitivity analysis~\cite{Marinescu09,Rouco01,Borsche15}, sliding mode controller design~\cite{Liao17}, the use of reference models~\cite{Yagooti16}, coordinated switching controllers~\cite{Liu16}, genetic algorithms based tuning~\cite{Taranto99}, model predictive control~\cite{Fuchs14}, and time-discretization~\cite{lei2001optimization}. An overview of different methods for power oscillation damping can, for example, be found in~\cite{obaid2017power}.
However, most of the works either: consider simplified power system models~\cite{Borsche15,MahmoudiNAPS15,Liao17}; or add and design new controllers on top of the existing power system model~\cite{Fuchs14,Raoufat16,Pipelzadeh17,Liu16,Preece13,wu2016input,Schuler14,Zhu03}. These solutions require significant modification of existing control structures, which makes practical application complex and expensive.
Very few publications consider the optimization of existing controller parameters~\cite{befekadu2005robust,marinescu2019residue,kammer2017decentralized}. The approaches in these works use heuristics~\cite{marinescu2019residue}, or assume a specific dependency on the parameters~\cite{befekadu2005robust,kammer2017decentralized}.

Controller synthesis based on $\Hinf$ optimization has received significant attention in the last decades.
First approaches in the 1980s used algebraic Riccati equations for the $\Hinf$ controller synthesis~\cite{doyle1989state}. In the 1990s, approaches based on linear matrix inequalities became popular, leading to convex solutions for unstructured state-feedback controller synthesis based on the bounded-real lemma~\cite{gahinet1994linear}. If the controller structure is fixed and only the parameters are tuned using $\Hinf$ optimization, as is the case in this work, one refers to structured $\Hinf$ controller design~\cite{Scherer13,apkarian2018structured}. 
Structured controller synthesis, exploiting the bounded real lemma and additional improvements and refinements are, e.g., used in~\cite{Schuler14, Hassibi1999,ibaraki2001rank,dinh2012combining,Han04,Karimi07,schuler2011design,befekadu2006robust,Mesanovic18ACC}.
Alternative tuning approaches, as non-smooth optimization~\cite{HIFOO,apkarian2006nonsmooth}, bisection~\cite{kanev2004robust} etc. exist. 

In recent years, the focus in structured $\Hinf$ optimization shifted towards more efficient methods to find local minima, as local solutions are often sufficient in practice. These methods are usually based on frequency sampling, leading to fast and reliable synthesis~\cite{kammer2017decentralized,apkarian2018structured,boyd2016mimo}. This, however, removes the guarantee that a stable controller will be obtained. To solve this issue,~\cite{kammer2017decentralized,apkarian2018structured} introduce stability constraints based on the Nyquist criterion. 
In~\cite{boyd2016mimo}, the assumption is made that the controlled plant is asymptotically stable, in which case the boundedness of the $\Hinf$ norm of the system sensitivity matrix is a necessary and sufficient condition for stability.

\begin{figure}[tb]
	\centering
	\includegraphics[width=1\columnwidth]{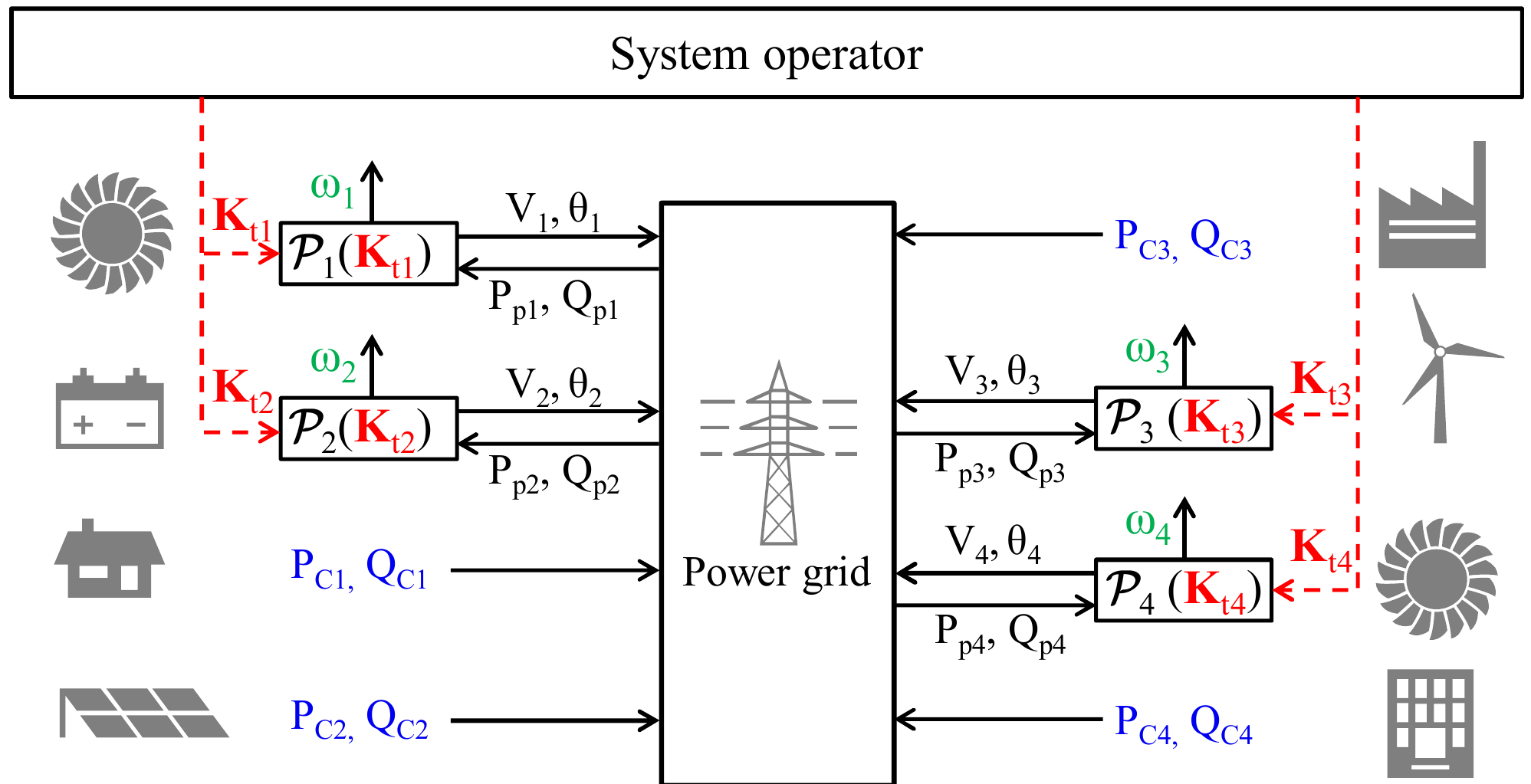}
	\caption{Exemplary power system consisting of four dynamic prosumers $\cP_i$ and four static prosumers $P_{ci}$, $Q_{ci}$. The tunable controller parameters $\vK_{ti}$ of dynamic prosumers are marked red. The static prosumers, marked with blue, are disturbance inputs into the system. The frequencies $\omega_i$ of each $\cP_i$ are performance outputs, marked green. For details, see Section~\ref{sec.PSModel}.}
	\label{fig.GridModel}
\end{figure}

We introduce in this work a stability guarantee for structured $\Hinf$ optimization based on frequency sampling, without adding new constraints in the optimization problem. 
Doing so, we do not require the stability of the open loop system, as is often necessary~\cite{boyd2016mimo}. Previous works considered structured controller synthesis with application to power systems~\cite{Mesanovic17ISGT,Mesanovic18ACC,mesanovic18ACDC} based on the application of the bounded real lemma. 
{The main contributions of this work are:
\begin{itemize}
	\item A structured $\Hinf$ controller tuning method which allows nonlinear parameter dependencies and which is scalable towards large scale systems.
	\item A stability certificate for the developed method.
	\item Two numerical studies on power transmission systems with 190 and 469 states, showing the applicability of the proposed approach.
	\item An evaluation of the approach on a testbed microgrid consisting of six inverters and a load-bank, showing the practical applicability of the approach on existing hardware.
	\item  A comparison of the approach to other $\Hinf$ tuning methods, which underlines the significantly faster computation times.
\end{itemize}
}

The remainder of this work is organized as follows:
Section~\ref{sec.PSModel} outlines the problem and derives suitable models.
We introduce the proposed structured controller synthesis approach with the stability proof in Section~\ref{sec.TuningMethod}. Section~\ref{sec.NumericalEval} presents the simulation studies considering power systems with 10 and 53 power plants. The experimental validation is presented in Section~\ref{sec.ExpEval}, {whereas Section~\ref{sec.PerfComparison} shows the computational performance of the approach} before providing conclusions in Section~\ref{sec.Conclusion}.


\subsection{Mathematical preliminaries}
\label{subsec.MathPrelim}
$\bigSigma(\cdot)$ denotes in the following the largest singular value of a matrix, equivalently $\overline{\lambda} (\cdot)$ denotes the largest eigenvalue of a  matrix, whereas $(\cdot)^*$  denotes the conjugate transpose of a matrix. The notation $\succ$ ($\succeq$), and $\prec$ ($\preceq$) is used to denote positive (semi)definiteness and negative (semi)definiteness of  a matrix, respectively. We use $j$ to denote the imaginary unit, $\R_{\ge 0}$ denotes the set of non-negative real numbers, $\Comp$ denotes the set of complex numbers, and $\Comp_{>0}$ denotes the set of complex numbers with a positive real part. The notation $\cR\Hinf$ is used to denote the set of proper, real rational and stable transfer matrices. {Furthermore, we define the operators $< (\le)$, $> (\ge)$, and $|\cdot|$ element-wise for vectors. }

\begin{defn}\cite{lunze2013regelungstechnik}\label{def.MIMOPoles}
	A complex number $s_p$ is a pole of the transfer matrix $G(s) : \Comp \rightarrow \Comp^{n_y \times n_w}$, when at least one element $G_{ij}(s)$ of $G(s)$ has a pole at $s_p$.
\end{defn}
We will reference the Bounded real Lemma, which states the following
\begin{lem} (Bounded-real Lemma) \label{lemma.BRL}
	\cite{gahinet1994linear} Consider the continuous-time transfer function $G(s)$ with the realization $G(s) = D + C (sI - A)^{-1} B$. 
	The following statements are equivalent
	\begin{itemize}
		\item The system $G(s)$ is asymptotically stable and $\left\| G(s) \right\|_\infty < \gamma$.
		\item There exits a symmetric positive definite solution $P \succ 0$ (Lyapunov matrix) to the linear matrix inequality (LMI)
		\begin{align}
		\begin{pmatrix}
		A^T P + P A & P B & C^T \\
		B^T P & -\gamma I & D^T\\
		C & D & -\gamma I
		\end{pmatrix} \prec 0. \label{eq.BRLConstraint}
		\end{align}
	\end{itemize}
\end{lem}

\section{Optimal Controller Tuning for Power Systems}
\label{sec.PSModel}
Figure~\ref{fig.GridModel} shows an exemplary power system with the basic idea of retuning controller parameters. 
It consists of heterogeneous components, such as power plants, renewable generation, storage systems and households.
We name these components prosumers, as they can either produce or consume electric power. Thereby, we distinguish between dynamic and static prosumers. 

Dynamic prosumers, such as power plants, are dynamic systems with internal states, denoted with $\cP_i$. They posses structured controllers, whose parameters $\vK_{ti}$ can be tuned, marked with red in Fig.~\ref{fig.GridModel}. We consider dynamic prosumers $\cP_i$ which control their voltage magnitude $V_i$ and phase $\theta_i$ at the point of connection, whereas their power infeed into the grid $P_{pi}$ and $Q_{pi}$ is the external input for the controllers. This is a standard description, e.g. for conventional power plants with synchronous generators~\cite{kundur93a}, as depicted in Fig.~\ref{fig.GridModel}, where $V_i$ and $\theta_i$ are outputs of $\cP_i$, and $P_{pi}$, $Q_{pi}$ are the inputs. However, the applicability of the modeling is not restricted to this dynamic prosumer type, and dynamic prosumers which have $P_{pi}$ and $Q_{pi}$ as output can also be considered.

Static prosumers, such as loads and some renewable generation, have no internal states and are characterized through their active and reactive power infeed, denoted with $P_{ci}$ and $Q_{ci}$, respectively. Figure~\ref{fig.GridModel} depicts four static prosumers, marked with blue. We collect infeeds of static prosumers into vectors $\vP_s$ and $\vQ_s$, which are considered as external inputs. Renewable generation and loads are often modeled as static prosumers~\cite{poolla2019placement,pddotnuschel2018frequency}. The power infeed of these elements cannot be fully controlled, and we consider these infeeds as the disturbance inputs for the controller tuning. Static prosumers also model components with a slow dynamic behavior, such as aggregated powers of small prosumers. For this reason, a subset of $\vP_{s}$ and $\vQ_{s}$ is chosen as the disturbance input $\vw_s$. The voltage phasors of buses with static prosumers have a magnitude $V_{si}$ and angle $\vtheta_{si}$.
Static and dynamic prosumers are coupled through the power grid.

Depending on the infeed of renewable generation and load, the system dynamic behavior changes. If the system operator thereby notices that the resiliency of the system decreases, it tunes the controller parameters $\vK_{ti}$ of dynamic prosumers in order to increase the system resiliency. The reparameterization process is depicted with red dashed lines in Fig.~\ref{fig.GridModel}. Thereby, slow communication is needed.

In power systems, the frequencies of the dynamic prosumers, defined with $\omega_i = \dot{\theta}_i$, where $\theta_i$ is the angle of the voltage phasor of $\cP_i$, are important and these are typically used to asses the system performance~\cite{kundur93a}. Thus, choosing the vector of frequencies as the performance output
\begin{align}
\vy = \begin{pmatrix} \omega_1 & ... & \omega_N	\end{pmatrix}^T, \label{eq.perfOutInit}
\end{align}
is a sensible choice. Here $N$ denotes the number of dynamic prosumers, and $\omega_i$ is the voltage frequency of $\cP_i$. The performance outputs are marked green in Fig.~\ref{fig.GridModel}.

In the following subsections, we outline the structure and dynamics of the power grid and prosumers and present possible models.


\subsection{Power grid}
\label{subsec.PowerGridModel}

The power grid consists of power lines, cables, transformers etc. which interconnect dynamic and static prosumers. In principle, each power line and cable, 
has its own dynamics, which, however, have time constants which are orders of magnitude smaller than the generation dynamics relevant for stability studies, 
which are often slower than 10 Hz~\cite{kundur93a}. For this reason, the power grid dynamics are often neglected~\cite{schiffer2016survey,kundur93a}. Consequently, the grid, i.e. the power flow, is typically described by the algebraic power flow equations
\begin{subequations} 
	\label{eq.PowerFlow}
	\begin{align}
	&P_i = \sum_{j=1}^{N_B} \nolimits V_{Bi} V_{Bj}\big( G_{cij}\cos \Delta\theta_{Bij}+B_{sij}\sin \Delta\theta_{Bij} \big)\label{eq.activepower} \\
	&Q_i = \sum_{j=1}^{N_B} \nolimits V_{Bi} V_{Bj} \big( G_{cij}\sin \Delta\theta_{Bij}-B_{sij}\cos \Delta\theta_{Bij}  \big),\label{eq.reactivepower}
	\end{align}
\end{subequations}
where $N_B$ is the number of buses (nodes) in the power system and is equal to the total number of dynamic and static prosumers in the grid, $P_i$ and $Q_i$, are the injected active and reactive powers into the i-th bus (node) in the grid by a dynamic prosumer ($P_{pi}$, $Q_{pi}$) or a static prosumer ($P_{si}$, $Q_{si}$),
$V_{Bi}$ and $\theta_{Bi}$ are the magnitude and angle of the voltage phasor at the i-th bus from a dynamic prosumer ($V_i$, $\theta_i$) or a static prosumer ($V_{si}$, $\theta_{si}$),
and $G_{cij}$ and $B_{sij}$ are the elements of the  conductance and susceptance matrix of the grid \cite{kundur93a}. 

{
As the share of fast, inverter-based generation in power systems increases, algebraic modeling of the power flow is facing increasing scrutiny. Fast control loops of inverters may cause unwanted interactions with power grid dynamics, raising the need for dynamic power flow equations~\cite{markovic2019understanding}. However, the considered numerical examples, and the testbed microgrid system did not have this issue. Including dynamic power flow equations into the modeling is part of future work. 
}



\subsection{Dynamic prosumers and tunable parameters}
\label{subsec.ProsumerModel}

The proposed structure allows for arbitrary dynamic prosumers. In this section, two exemplary prosumers and their models are outlined. 

\begin{figure}[tb]
	\centering
	\includegraphics[width=0.7\columnwidth]{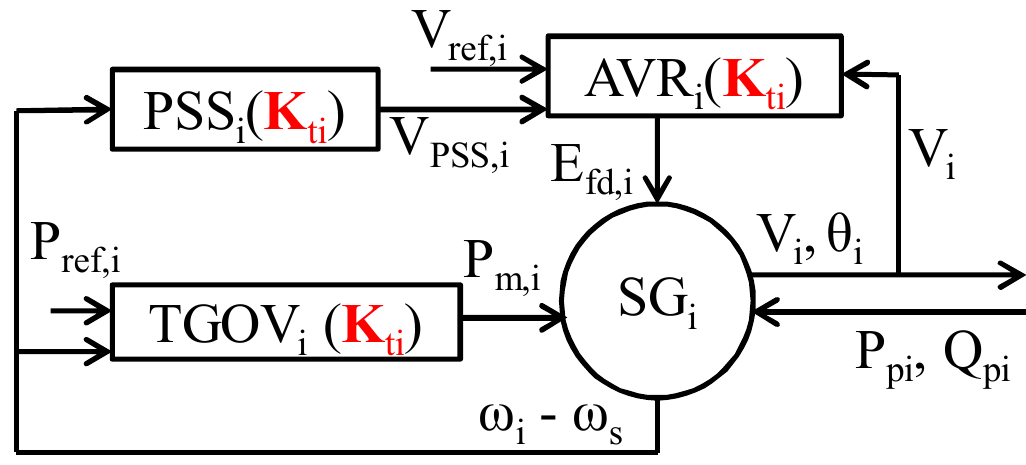}
	\caption{Simplified model of a dynamic prosumer $\cP_i$, a power plant. It consists of a synchronous generator (SG$_i$), automatic voltage regulator and exciter (AVR$_i$), power system stabilizer (PSS$_i$), and of a turbine and governor model (TGOV$_i$).}
	\label{fig.SGModel}
\end{figure}

\begin{figure}[tb]
	\centering
	\includegraphics[width=0.8\columnwidth]{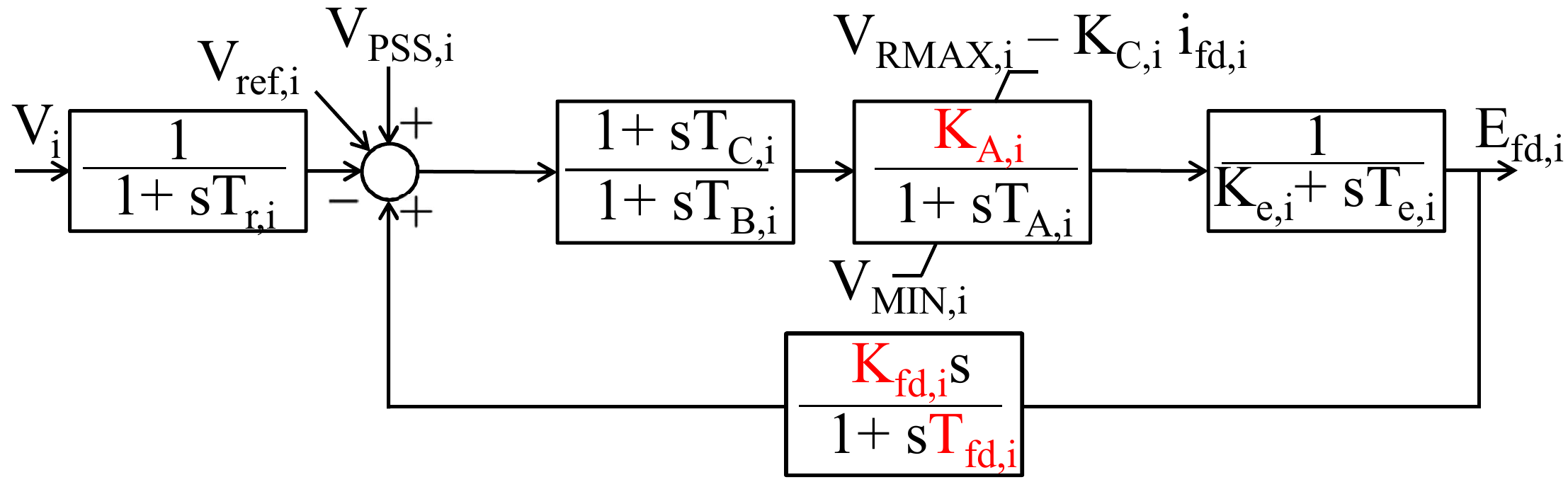}
	\caption{{Dynamic model of AVR$_i$~\cite{MathworksExciter}, where $T_{r,i}$ is the transducer time constant, $T_{C,i}$ and $T_{B,i}$ are dynamic gain reduction time constants, $K_{A,i}$ is the AVR gain, $T_{A,i}$ is the AVR lag time constant, $K_{e,i}$ and $T_{e,i}$ are the exciter parameters, and $K_{fd,i}$ and $T_{fd,i}$ additional damping coefficients of the AVR. We assume that $K_{A,i}$, $K_{fd,i}$, and $T_{fd,i}$, marked red, are tunable.}}
	\label{fig.Exciter}
\end{figure}

\begin{figure}[t]
	\centering
	\includegraphics[width=0.9\columnwidth]{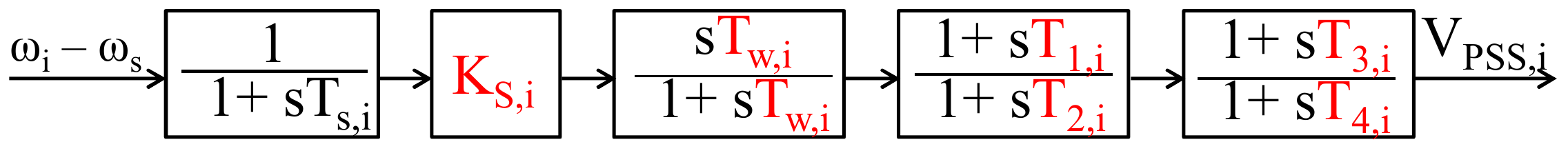}
	\caption{{Dynamic model of the simple power system stabilizer (taken from~\cite{moeini2015open,kundur93a}), where $K_{S,i}$ is the PSS gain, $T_{w,i}$ is the washout time constant, $T_{1,i}$-$T_{4,i}$ are the lead-lag filters time constants, and $T_{s,i}$ is the sensor time constant. All of the PSS parameters are tunable, except the sensor time constant.}}
	\label{fig.PSS}
\end{figure}

\begin{figure}[t]
	\centering
	\includegraphics[width= 1\columnwidth]{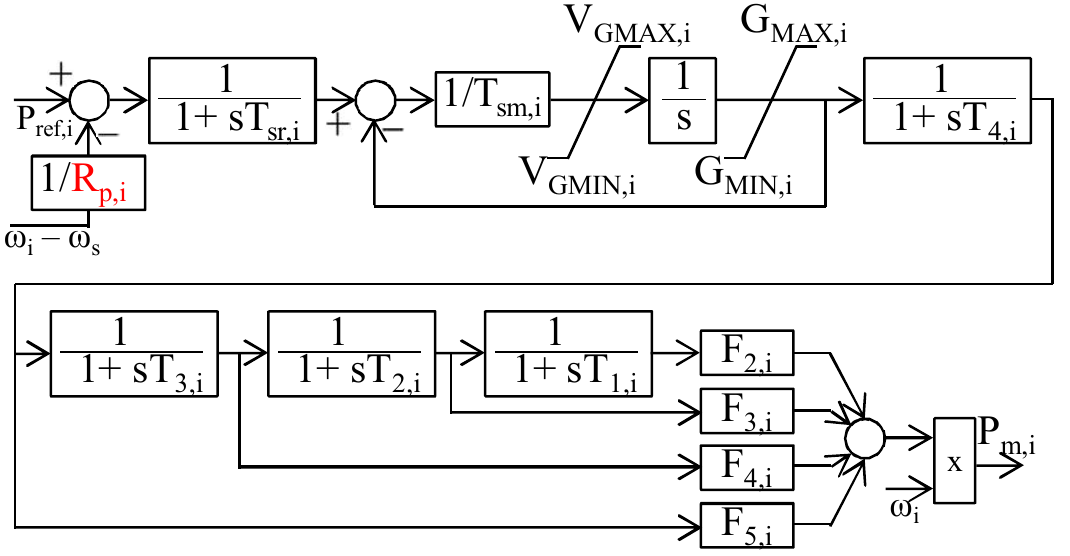}
	\caption{{Dynamic model of the turbine and governor from~\cite{TGOVMathworks}. The frequency droop gain of the governor $R_{p,i}$ is an optimization variable.}}
	\label{fig.TGOV}
\end{figure}

\subsubsection{Power plants}

Power plants	 often consist of a synchronous generator with controllers and actuators, as shown in Fig.~\ref{fig.SGModel}. 
We consider the 6-th order model for the synchronous generator (SG$_i$). 
For details, we refer the interested readers to~\cite{kundur93a}.

The automatic voltage regulator and exciter (AVR$_i$) represents hardware and controllers which control the voltage at the power plant terminals $V_{i}$ as close as possible to a reference value $V_{ref,i}$. The output of AVR$_i$ is the field winding voltage $E_{fd,i}$ which is an input of the SG$_i$. {An exemplary model of an AVR with three tunable parameters, marked red, is shown in Fig.~\ref{fig.Exciter}.}
It is important to note that automatic voltage regulators can reduce the stability margin in power systems~\cite{kundur93a}. For this reason, power plants are sometimes equipped with power system stabilizers (PSS$_i$). PSSs are analogue or digital controllers, with the task to improve the system stability and increase the damping of oscillations in power systems.
{An exemplary PSS is visualized in Fig.~\ref{fig.PSS}.}
The governor and turbine (TGOV$_i$) control the generator frequency by adapting the mechanical power $P_{m,i}$ transferred to the synchronous generator, see Fig.~\ref{fig.TGOV} for an example of a TGOV model with one tunable parameter, marked red.
{
In practice, many different controllers are used, see e.g.~\cite{IEEEExciters06}. Examples for these controllers are shown in Figs.~\ref{fig.Exciter}-\ref{fig.TGOV}, which are used in a subsequent numerical study.} 

\subsubsection{Inverters}

\begin{figure}[tb]
	\centering
	\includegraphics[width=0.95\columnwidth]{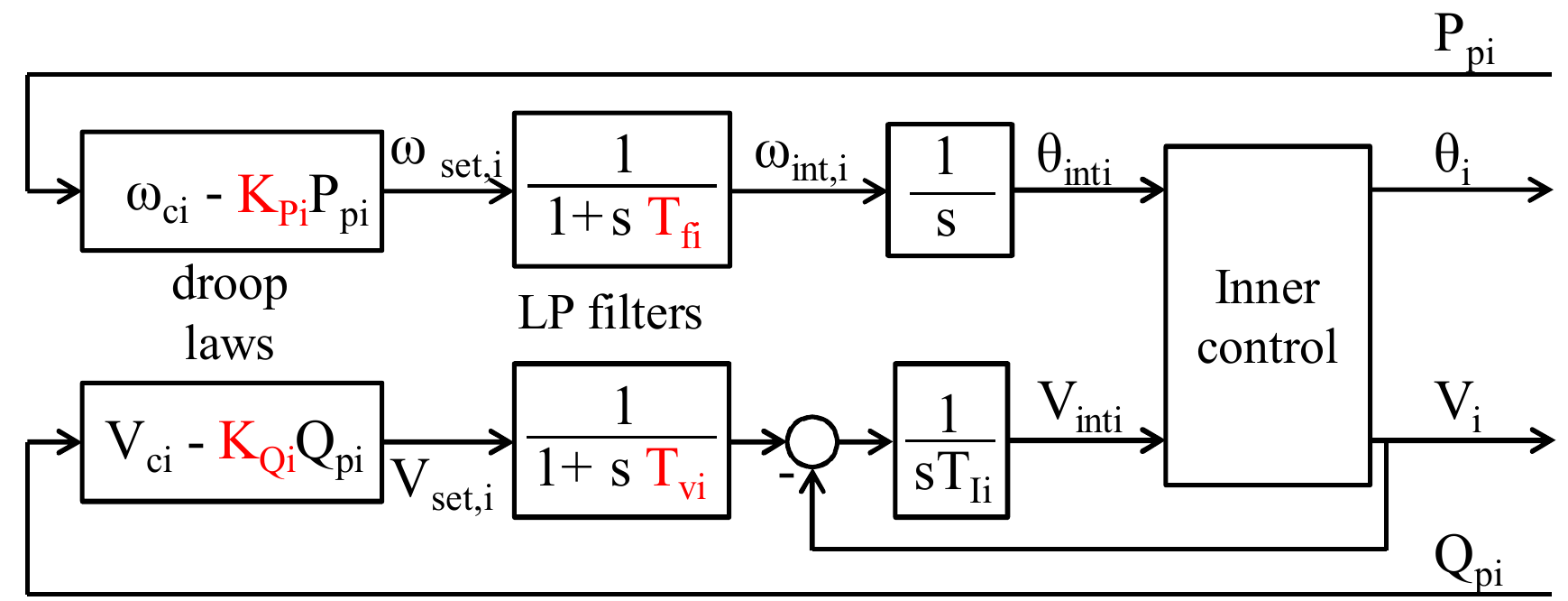}
	\caption{Simplified model of a dynamic prosumer $\cP_i$, an inverter with so-called droop controls.}
	\label{fig.InvModel}
\end{figure}

We consider inverters which control the voltage and frequency at their terminals, called voltage-source inverters (VSI), or inverters in grid-forming mode. 
For dynamics below 10 Hz, modeling the high-frequency switching of power transistors in the inverters is often not necessary. Instead, the transistors are approximated as ideal voltage sources with droop controllers for voltage amplitude and frequency, c.f. Fig.~\ref{fig.InvModel}.
The DC link capacitor of the inverters is not considered, as we assume that the internal control of the inverters is fast to compensate for the changes on the DC side. Such simplifications comply with measurements shown in~\cite{rahmoun2017mathematical}, and with the experiments considered in Section~\ref{sec.ExpEval}.
{Note that the presented inverter model, representing the behavior of a SINAMICS inverter~\cite{sinamics}, does not have fast inner control in a synchronously rotating frame~\cite{markovic2019understanding}.}

In grid-forming mode, the i-th inverter controls the magnitude $V_i$ and phase $\theta_i$ of the voltage on its terminals, whereas the active and reactive power infeed of the inverter result from the power flow. The frequency setpoint of the inverter $\omega_{seti}$ is determined by the so-called droop equation
\begin{align}
& \omega_{seti} = \omega_{ci} - K_{Pi} P_{pi},\label{eq.freqDroop}
\end{align}
where $\omega_{ci}$ is the frequency setpoint with zero load, $P_{pi}$ is the measured active power infeed of the i-th inverter, and $K_{Pi}$ is the frequency droop gain. The setpoint $\omega_{seti}$ is filtered with a first-order low-pass filter with the time constant $T_{fi}$ and integrated to obtain the internal voltage phase $\theta_{inti}$. Analogously, the voltage setpoint $V_{seti}$ is determined with the so-called droop equation
\begin{align}
V_{seti} = V_{ci} - K_{Qi} Q_{pi},
\end{align}
where $V_{ci}$ is the voltage setpoint with no reactive power generation, $Q_{pi}$ is the measured reactive power infeed of the i-th inverter, and $K_{Qi}$ and is the frequency droop gain. The setpoint $V_{seti}$ is filtered with a time constant $T_{vi}$, and serves as the setpoint for the integral voltage controller. The output of the integral controller is the internal voltage $V_{inti}$.

The resulting $\theta_{inti}$ and $V_{inti}$ are used as references to the internal control loops which run at a much higher frequency. As the internal control loops are not modeled due to their fast dynamics, we assume $\omega_{i} = \omega_{inti}$, $\theta_{i} = \theta_{inti}$ and $V_{i} = V_{inti}$.

The tunable inverter parameters are marked red in Fig.~\ref{fig.InvModel}, they are: $\vK_{ti} = (K_{Pi}, K_{Qi}, T_{fi}, T_{vi})^T$. {Note that we only tune parameters of the outer control loops, operating on a slower time scale than e.g. current controllers. This, together with the fact that $T_{fi}, T_{vi}$ are restricted to values above 50ms, prevents undesired interactions between inverter control and grid dynamics, enabling the use of algebraic power flow equations in the considered testbed.}



\subsection{Overall problem setup}
\label{subsec.CoupledModel}

Combining the power grid equations~\eqref{eq.PowerFlow} with the prosumer models leads to differential-algebraic nonlinear equations
\begin{subequations} \label{eq.nonlinearModel}
	\begin{align}
	\dot{\vx} =& f(\vx,\vw,\vK_t) \\
	0 =& h(\vx,\vw,\vK_t), \label{eq.nonlinearModel.alg}
	\end{align}
\end{subequations}
where $\vx \in \R^{\cdot N_x}$ combines all dynamic prosumer states,
$\vw \in \R^{n_D}$ is the vector of disturbance inputs, represented by a subset of $\vP_C$ and $\vQ_C$, 
$\vK_t \in \R^{N_t}$ is the vector of tunable controller parameters of all dynamic prosumers,
$f$ describes the prosumer dynamics,
and $h$ represents the power flow equation~\eqref{eq.PowerFlow}. We combine the system dynamics~\eqref{eq.nonlinearModel} with the performance output~\eqref{eq.perfOutInit} to obtain the nonlinear model of our power system.

{Our goal is to increase of system robustness towards disturbances from $\vw$ (a subset of $\vP_c$ and $\vQ_c$) and to improve the time domain behavior of the system.}
For simplicity, we linearize~\eqref{eq.nonlinearModel} around a known steady-state $\vx_0$ with the known input $\vw_0$. While this is an approximation, it allows us to use linear systems methods. It has furthermore been shown to be sufficient even for large-scale disturbances~\cite{poolla2019placement}.
{Moreover, if large disturbances are considered, several linearization points may be used}. After eliminating the linearized algebraic equation~\eqref{eq.nonlinearModel.alg}, we obtain the overall system
	\begin{align} \label{eq.linearizedModel1}
	\dot{ \widetilde \vx} &= \widetilde A(\vK_t) \widetilde \vx +  \widetilde B (\vK_t) \vw \qquad
	\vy  =\widetilde  C \: \widetilde \vx.
	\end{align}
The system matrix $\widetilde A(\vK_t)$ has an eigenvalue at zero, as the coupling power flow equation~\eqref{eq.PowerFlow} is invariant under phase offsets $\widetilde \theta_i = \theta_i + \delta \theta$, where $\delta  \theta \in \R$ is identical for all $i$. This zero eigenmode can be eliminated~\cite{wu2016input}, leading to
	\begin{align}  \label{eq.linearizedModel}
	\dot{ \vx} &=  A(\vK_t)  \vx +   B (\vK_t) \vw \qquad 
	\vy  = C \vx,
	\end{align}
where $A(\vK_t)$ has no parameter-independent eigenvalues on the imaginary axis, and $A(\vK_t)$ has all eigenvalues in the left half plane for suitable choice of $\vK_t$. 
The resulting state space system can be written in the frequency domain as
\begin{align}
G(\vK_t,s) = C \left(sI - A(\vK_t)\right)^{-1} B(\vK_t).
\end{align}

\section{Method for structured $\Hinf$ parameter tuning}
\label{sec.TuningMethod}

This section defines an optimization algorithm which minimizes the $\Hinf$ norm of $G(\vK_t,s)$, denoted with $\|G(\vK_t,s) \|_\infty$, and defined by~\cite{boyd1985subharmonic}
\begin{align}
\|G(s)\|_\infty \DefinedAs  \text{sup}_{s \in \Comp_{>0}} \: \bigSigma \: (G(s)) &= \text{sup}_{\omega \in \mathbb{R}} \: \bigSigma \: (G(j \omega)). \label{eq.HinfNormDefinition}
\end{align}
Note that the last equality is only valid for stable systems. The $\Hinf$ norm is chosen, as it represents the maximal amplification of amplitude of any harmonic input signal in any output direction. Thus, minimizing the $\Hinf$ norm minimizes the worst-case amplification of oscillation frequencies after a disturbance. 
Thereby, the system robustness is additionally improved.
Minimization of $\|G(\vK_t, s)\|_\infty$ is achieved by optimizing the vector of tunable parameters $\vK_t$. For notational convenience, we drop writing the dependency on the tunable parameter vector $\vK_t$ explicitly, when it is not necessary, and write only $G(s)$. It is assumed, however, that $G(s)$ is always a function of the tunable parameters.

We consider that the system $G(\vK_t,s) \in \cR\Hinf$ is  an asymptotically stable and detectable linear time-invariant multi input multi output (LTI MIMO) system. 
It has a nonlinear dependency on the vector of tunable controller parameters $\vK_t$. 

The basis for the proposed parameter tuning method is the following theorem, which can be found in the literature in analogous forms, see e.g.~\cite{boyd2016mimo,kammer2017decentralized} and references therein.
\begin{thm}[Semi-infinite $\Hinf$ constraint] \label{thm.InfinityConstr}
	Given a detectable and asymptotically stable system $G(s)$. The $\Hinf$ norm of $G(s)$ is smaller than $\gamma \in \R$ if and only if
	\begin{align} \label{eq.InfConstr}
	\begin{pmatrix}
	\gamma I     & G(j\omega) \\
	G(j\omega)^* & \gamma I
	\end{pmatrix} \succ 0, \quad \forall \omega \in \R_{\ge 0}.
	\end{align}
\end{thm}
\begin{proof}
	Since $\bigSigma \: (G(j \omega))^2 = \overline{\lambda} (G(j \omega)^* G(j \omega))$, it follows
	\begin{align}
	& \| G(j \omega) \|_\infty < \gamma \\
	\Leftrightarrow \quad &  \overline{\lambda} (G(j \omega)^* G(j \omega)) < \gamma^2, \quad \forall \omega \in \R_{\ge 0}\\
	\Leftrightarrow \quad &  G(j \omega)^* G(j \omega) - \gamma^2 I \prec 0, \quad \forall \omega \in \R_{\ge 0}.
	\end{align}
	By using the Schur complement on the last expression, we obtain~\eqref{eq.InfConstr}.
\end{proof}
Theorem~\ref{thm.InfinityConstr} allows to directly formulate an optimization problem for the $\Hinf$ minimization of $G(\vK_t,j \omega)$
\begin{subequations} \label{eq.InfOptProb}
	\begin{align}
	\min_{\gamma, \vK_t}  & \quad \gamma \\
	\text{s.t.} \quad & \begin{pmatrix}
	\gamma I     & G(\vK_t, j\omega) \\
	G(\vK_t, j\omega)^* & \gamma I
	\end{pmatrix} \succ 0, \quad \forall \omega \in \R_{\ge 0} \label{eq.infiniteConstraint} \\
	& \vK_{tmin} \leq \vK_t \leq \vK_{tmax}. \label{eq.infProb.ContConstr}
	\end{align}
\end{subequations}
The last inequality is a box constraint on the controller parameters, determined by practical considerations or physical constraints.
As~\eqref{eq.infiniteConstraint} needs to be satisfied for every $\omega \in \R_{\ge 0}$, Problem~\eqref{eq.InfOptProb} is semi-infinite.
This formulation is similar to those considered in~\cite{boyd2016mimo,kammer2017decentralized,apkarian2018structured} and the references therein.
One way to find a finite-dimensional approximation to~\eqref{eq.InfOptProb} is to use a finite, but large enough, number of frequency samples at which constraint~\eqref{eq.infiniteConstraint} is evaluated
\begin{subequations} \label{eq.FiniteOptProb}
	\begin{align}
	\min_{\gamma, \vK_t}  & \quad \gamma \\
	\text{s.t.} \: \: & \Phi(G(\vK_t, j\omega_\mu), \gamma) = \begin{pmatrix}
	\gamma I     & G(\vK_t, j\omega_\mu) \\
	G(\vK_t, j\omega_\mu)^* & \gamma I
	\end{pmatrix} \succ 0, \nonumber \\
	& \qquad \qquad \qquad \forall \omega_\mu \in \Omega \label{eq.finiteConstraints} \\
	& \vK_{tmin} \leq \vK_t \leq \vK_{tmax}.
	\end{align}
\end{subequations}
Here $\Omega$ is the discrete set of sampled frequencies with $N_\omega$ elements. Since the problem scales linearly with $N_\omega$, a reasonably large number of elements in $\Omega$ can be chosen such that it covers the required frequency range with satisfactory density~\cite{boyd2016mimo}. Note that the choice for $\Omega$ is problem specific and needs to be adapted to the considered frequency range.
With a sufficiently large number of samples in $\Omega$, the local optimum of~\eqref{eq.FiniteOptProb} can be arbitrarily close to the optimum of~\eqref{eq.InfOptProb}.
The advantage of~\eqref{eq.FiniteOptProb} compared to methods based on Lemma~\ref{lemma.BRL}, with respect to scalability, are severalfold. Approaches based on Lemma~\ref{lemma.BRL} introduce a positive-definite (Lyapunov) matrix as an optimization variable, which has the same size as the closed loop system, causing the number of optimization variables to increase quadratically with the number of states in the closed-loop system. Additionally, the size of the matrix in~\eqref{eq.BRLConstraint} scales linearly with the number of states, inputs, and outputs of the system. Problem~\eqref{eq.FiniteOptProb} does not have the Lyapunov matrix $P$ as an optimization variable, and the size of the problem only depends on the number of inputs and outputs of the system, making the controller synthesis generally faster. 

Problems~\eqref{eq.InfOptProb} and~\eqref{eq.FiniteOptProb}, however, do not guarantee system stability in a straightforward manner, i.e. a controller parameterization obtained as a solution of~\eqref{eq.FiniteOptProb} does not necessarily stabilize the system in addition to minimizing the cost function representing the $\Hinf$ norm. To overcome this problem, one can introduce constraints based on the Nyquist criterion, c.f.~\cite{apkarian2018structured,kammer2017decentralized}, which guarantee closed-loop stability. If the open-loop system, i.e. the system without controllers, is stable, then the boundedness of the $\Hinf$ norm of the system sensitivity matrix ensures the stability of the closed-loop system~\cite{boyd1991linear,boyd2016mimo}.
\begin{rem}
	Note that, even though boundedness of the system $\Hinf$ norm is a necessary and sufficient condition for system stability,~\eqref{eq.FiniteOptProb} does not guarantee the synthesis of a stable controller parameterization. This is because the last equality in~\eqref{eq.HinfNormDefinition} is applicable if and only if $G(s)$ is exponentially stable~\cite{boyd1985subharmonic}. We overcome this by providing a suitable stability certificate for the solution of Problem~\eqref{eq.FiniteOptProb}.
\end{rem}

To this end, we introduce two lemmas and propose a theorem for the stability certificate.
\begin{lem} \label{lem.PoleLimitSV}
	Given a detectable and exponentially stable system $G(\vK_t,s)$ with a fixed parameter vector $\vK_t$ and the finite set of poles $\cS_H$. The largest singular value of $G(s)$, denoted with $\bigSigma(G(s))$,  approaches $+\infty$ as $s$ approaches any $s_{pij} \in \cS_H$, where $s_{pij}$ denotes the $p$-th pole of the transfer function in the $i$-th row and $j$-th column of $G(s)$.
\end{lem}
\begin{proof}
	For clarity of presentation, we present the proof when $s_{pij}$ is a pole with a multiplicity of one. The proof when $s_{pij}$ is a pole with larger multiplicity is analogous.
	Per definition, we have~\cite{zhou1998essentials}
	\begin{align}
	\bigSigma(G(s)) = \max_{\|\vz\|_2 = 1} \| G(s) \vz \|_2.
	\end{align}
	Thus, for all $s_{pij} \in \cS_H$
	\begin{align}
	\lim_{s \rightarrow s_{pij}} \bigSigma(G(s)) & = \lim_{s \rightarrow s_{pij}} \max_{\|\vz\|_2 = 1} \| G(s) \vz \|_2 \nonumber \\
	& \ge \lim_{s \rightarrow s_{pij}}  \| G(s) e_j \|_2, \label{eq.SigmaProof} 
	\end{align}
	where $e_j$ denotes a column vector where the $j$-th row is equal to one and all other elements are zero. The last expression can be reformulated to
	\begin{align}
	&\lim_{s \rightarrow s_{pij}}  \| G(s) e_j \|_2 = \lim_{s \rightarrow s_{pij}}  \left\| \Big(
	G_{1j}(s) \: ...\: G_{ij}(s) \:...\: G_{Nj}(s) \Big)^T  \right\|_2 \nonumber \\
	& = \lim_{s \rightarrow s_{pij}} \sqrt{G_{1j}^2(s) + ... + G_{ij}^2(s) + ... + G_{Nj}^2(s)}, \label{eq.SigmaProofVecNorm}
	\end{align}
	where $G_{ij}(s)$ denotes the single-input-single-output (SISO) transfer function in the i-th row and j-th column of $G(s)$.
	Since $s_{pij}$ is a pole of $G_{ij}(s)$, it follows that $\lim_{s \rightarrow s_{pij}} G_{ij}(s)^2 = +\infty$ and that $\lim_{s \rightarrow s_{pij}}  \| G(s) e_j \|_2 = +\infty$. From~\eqref{eq.SigmaProof}, it directly follows that $\lim_{s \rightarrow s_{pij}} \bigSigma(G(s)) = +\infty$. 
\end{proof}

\begin{lem} \label{lem.ContPoles}
	Given the linear system $G(\vK_t,s)$.
	If the denominator polynomials in $G(\vK_t, s)$ are continuous functions of the controller parameters $\vK_t$, then the location of poles of $G(\vK_t, s)$ are also continuous functions of the controller parameters $\vK_t$.
\end{lem}
\begin{proof}
	According to Definition~\ref{def.MIMOPoles}, the poles of $G(\vK_t, s)$ are obtained as the roots of denominator polynomials of all elements $G_{ij}(\vK_t,s)$ of $G(\vK_t,s)$. The roots of a polynomial are continuous functions of the polynomial coefficients~\cite{uherka1977continuous}, whereas the denominator polynomial coefficients are continuous functions of the controller parameters. It follows that poles of $G(\vK_t,s)$ are continuous functions of $\vK_t$.
\end{proof}
\begin{rem}
	In Lemma~\ref{lem.ContPoles}, we make the assumption that denominator polynomials in $G(\vK_t, s)$ are continuous functions of the controller parameters $\vK_t$. This assumption is satisfied for almost all practically relevant control elements, such as PID controllers, notch filters, lead-lag filters, washout filters etc. Hence, this assumption does not introduces a significant restriction.
\end{rem}


We can now formulate a stability certificate to validate that the closed-loop is stable.
{
\begin{thm}[Semi-infinite stability certificate] \label{thm.StabilityInfinite}
	Given an initial, exponentially stabilizing parameterization $\vK_{t,0}$ for the detectable system $G(\vK_t,s)$ with the set of poles $\cS_H$. Furthermore, using an iterative solution algorithm which minimizes the value of the cost function of the semi-infinite Problem~\eqref{eq.InfOptProb} in each iteration, thereby producing the parameter vector $\vK_{t,k}$ in the k-th iteration.
	Then there exists a sufficiently small step size $\vDK$,  $\vDK \ge |\vK_{t,k} - \vK_{t,k-1}|$ such that the (local) optimum of~\eqref{eq.InfOptProb} exponentially stabilizes $G$. 
\end{thm}
}
\begin{proof}
	{
	We first show that the solution algorithm will not allow for poles on the imaginary axis, and afterwards we show that there exists a step size which prevents an unstable parameterization from occurring.}

{
	Since $\bigSigma(G(s))$ is a continuous function of $s$~\cite{de1989analytic}, and $\lim_{s \rightarrow s_{p} \in \cS_H} \bigSigma(G(\vK_{t},s)) = +\infty$ (Lemma~\ref{lem.PoleLimitSV}), there exists $\varphi_{s_{p}} \in \R_{> 0}$, such that for any $\vK_{t}, \vK_{tmin} \leq \vK_t \leq \vK_{tmax}$, the following relation holds
		\begin{align} \label{eq.varphidef}
			\! \! \! \!| s - s_{p}(\vK_t) | \leq \varphi_{s_{p}} \Rightarrow  \bigSigma(G(\vK_{t},s)) \ge \max_{\omega \in \R} \bigSigma(G(\vK_{t},j\omega)).
		\end{align}
		Here $\varphi_{s_{p}}$ defines a neighborhood of $s_{p}$ in which the value of $\bigSigma(G(\vK_{t},s))$ is greater than $\HinfNormG$ for any allowed $\vK_t$. Thus, in a (local) optimum, every $s_P \in \cS_H$ will have at least the distance $\varphi_{s_{p}}$ from the imaginary axis. Otherwise, $\HinfNormG$ would increase. 
		Furthermore, as $s_{p}(\vK_t)$ is continuous (Lemma~\ref{lem.ContPoles}), a maximal step size $\vDK$ exists such that
		\begin{align} \label{eq.vDKDef}
		 \vDK \ge |\vK_{t,k} - \vK_{t,k-1}| \Rightarrow
			| s_{p}(\vK_{t,k}) - s_{p}(\vK_{t,k-1}) | \leq \varphi_{s_{p}}.
		\end{align}
		With the previously defined $\vDK$, $s_{p}$ will never cross from the stable to the unstable region, as this would mean that the value $\left\|G(\vK_{t,k}, s) \right\|_\infty$ would increase in at least one iteration, in which $-\varphi_{s_{p}} \le \myRe(s_p(\vK_{t,k}))<0$, which is not allowed by the solver.
	}

\end{proof}
{
The calculation of the allowed step size vector $\vDK$ from~\eqref{eq.vDKDef} is numerically expensive and conservative. To avoid this step, we introduce the following algorithm: if an unstable parameterization $\vK_{t,k}$ is obtained, the allowed step size is iteratively multiplied by a scalar $\alpha \in (0,1)$ until a stable parameterization is obtained. This solution is used in the subsequently proposed algorithm. 
The previous theorem is valid for the semi-infinite problem~\eqref{eq.InfOptProb}. In the following theorem, we extend the stability-certificate to discrete sampling.
}

{
\begin{thm}[Discrete stability certificate] \label{thm.Stability}
Given an initial, exponentially stabilizing parameterization $\vK_{t,0}$ for the detectable system $G(\vK_t,s)$ with the set of poles $\cS_H$. Furthermore, using an iterative solution algorithm which minimizes the value of the cost function of the semi-infinite problem~\eqref{eq.FiniteOptProb} in each iteration, thereby producing the parameter vector $\vK_{t,k}$ in the k-th iteration.
Then there exists a sufficiently small (vector) step size $\vDK$,  $\vDK \ge |\vK_{t,k} - \vK_{t,k-1}|$ and a sufficiently big (dense) discrete set $\Omega$ such that the (local) optimum of~\eqref{eq.InfOptProb} exponentially stabilizes $G$. 
\end{thm}
\begin{proof}
	We only outline the proof due to page limitations. The basic idea is that if the density of sampling frequencies $\omega_\mu \in \Omega$ in the k-th iteration is smaller than $2\varphi_{s_p}$, defined with~\eqref{eq.varphidef}, in the frequency regions around relevant poles, then $\max_{\omega_\mu \in \Omega} \bigSigma(G(\vK_{t,k}),j\omega_\mu)$ will increase when a pole approaches the imaginary axis, preventing the pole to become unstable. The detailed proof can be obtained analogously as for Theorem~\ref{thm.StabilityInfinite}.
\end{proof}
}

A direct consequence of the previous theorems is that Problem~\eqref{eq.FiniteOptProb} cannot stabilize an unstable system. If the initial parameterization $\vK_{t0}$ is unstable, Problem~\eqref{eq.FiniteOptProb} will not allow unstable system poles to cross the imaginary axis to the stable region. Note that the requirement for an initial stabilizing controller is in accordance to the results presented in~\cite{apkarian2018structured,kammer2017decentralized}.
In comparison to~\cite{boyd2016mimo}, the stability guarantee in Theorem~\ref{thm.Stability} is applicable to systems which are open-loop unstable, whereas~\cite{boyd2016mimo} requires that the open-loop system is stable.
For better understanding of the claim of Theorem~\ref{thm.Stability}, a small example system is visualized in~\ref{App.ExampleSystem}.


Problem~\eqref{eq.FiniteOptProb} is non-convex due to the nonlinear dependency on the controller parameters in $G(\vK_t,s)$. In order to solve it with convex solvers, we transform the problem into a series of convex optimization problems by linearizing the parameter dependency of $G(\vK_t,s)$. To obtain the linearized transfer matrix in the k-th iteration $G_{L,k}(\vK_t,s)$, we linearize $G(\vK_t,s)$ around the parameter vector obtained in the previous iteration $\vK_{t,k-1}$. The following optimization problem is then solved in each iteration
\begin{subequations} \label{eq.ConvexOptProb}
	\begin{align}
	\min_{\gamma, \vK_{t,k}}  & \quad \gamma \\
	\text{s.t.} \quad & \Phi(G_{L,k}(\vK_{t,k}, j \omega_\mu),\gamma) \succ 0, \forall \omega_\mu \in \Omega\\
	& \vK_{tmin} \leq \vK_{t,k} \leq \vK_{tmax}\\
	& | \vK_{t,k} - \vK_{t,k-1}| \leq \vDK, \label{eq.TrustRegion}
	\end{align}
\end{subequations}
where $\Phi$ is defined in~\eqref{eq.FiniteOptProb}, and
we define the absolute value element-wise for vectors. 
Constraint~\eqref{eq.TrustRegion} has two purposes in the optimization algorithm. First, it defines a trust region in which the linearization accuracy in $G_{L,k}(\vK_{t,k})$ is preserved.
Secondly, by reducing $\vDK$, it can be used to reduce the step size if we obtain an unstable system during optimization.
The resulting iterative convex optimization algorithm is outlined in Fig.~\ref{fig.OptAlg}. 
{In Step~\ref{st.TRAdapt}, $\vDK$ is reduced through multiplication with $\alpha \in (0,1)$, which serves multiple purposes. Firstly, it is used to satisfy the assumptions of Theorem~\ref{thm.Stability}. It ensures that the rise of the system $\Hinf$ norm, due to the approach of a pole, is captured in the optimization. Secondly, it is used to reduce the ``trust region'' around the linearization point of the parameter dependency if it is not accurate enough, i.e. if $\left\| G(\vK_{t,k},s) \right\|_\infty$ increased, although $\left\| G_{L,k}(\vK_{t,k},s) \right\|_\infty$ was reduced.
In subsequent numerical evaluations the value $\alpha = 0.7$ is used. 
Step~\ref{st.SamplingAdapt} adapts the frequency grid $\Omega$ if it is not sufficiently dense, such that Theorem~\ref{thm.Stability} is applicable. To this end, frequencies of unstable poles are added to $\Omega$.}
By choosing $\vDK$ small enough and with sufficient sampling, Theorem~\ref{thm.Stability} guarantees that a stabilizing controller is obtained. Convergence to a local optimum is guaranteed if the initial value is close enough to the (locally) optimal value~\cite{nocedal2006numerical}. Note that the proposed optimization is applicable to arbitrary systems which satisfy the previous assumptions, and not only to electrical networks. {However, a large frequency grid $\Omega$ should be avoided to reduce computation times.}

{
Problem~\eqref{eq.ConvexOptProb} can be extended to optimize the controller parameters to several system realizations $\kappa$ by solving the following optimization problem
\begin{subequations} \label{eq.FiniteOptProbMultiScenario}
	\begin{align}
	\min_{\gamma, \vK_t}  & \quad \gamma \\
	\text{s.t.} \: \: & \Phi(G_\kappa(\vK_t, j\omega_\mu), \gamma) \succ 0, \: \forall \omega_\mu \in \Omega_\kappa, \: \kappa  = 1...N_s\\
	& \vK_{tmin} \leq \vK_t \leq \vK_{tmax},
	\end{align}
\end{subequations}
where $N_s$ is the total number of systems. Such formulations are necessary, e.g. when one parameterization is needed for several system realizations. Problem~\eqref{eq.FiniteOptProbMultiScenario} can be solved analogously with the algorithm in Fig.~\ref{fig.OptAlg} and is used subsequently in a numerical example.
}
\begin{figure}[tb]
	\begin{algorithmic}[1]
		\Procedure{StructHinfTuning}{$G$, $\vK_{t,0}$, $\vDK$, $k_{max}$}
		\State $k=1$, choose $0< \alpha < 1$
		\While{ $k\le k_{max}$ or not converged}
		\State $G_{L,k}(\vK_t)\gets$  linearize $G(\vK_t)$ around $\vK_{t,k-1}$
		\State $\vK_{t,k} \gets $ solution of~\eqref{eq.ConvexOptProb}.
		\If {$\|G(\vK_{t,k},s)\|_\infty \ge \|G(\vK_{t,k-1},s)\|_\infty$ or $G(\vK_{t,k},s)$ is unstable}
		\State $\vDK \gets \vDK \times \alpha$ \label{st.TRAdapt}
		\State Increase the frequency sampling if necessary. \label{st.SamplingAdapt}
		\State $\vK_{t,k} \gets \vK_{t,k-1}$
		\EndIf 
		\State $k \gets k+1$
		\EndWhile
		\EndProcedure
	\end{algorithmic}
	\caption{Proposed iterative parameter optimization algorithm.} 
	\label{fig.OptAlg}
\end{figure}

\section{Simulation Studies}
\label{sec.NumericalEval}

We evaluate the proposed method considering two power system models with 10 and 53 power plants, respectively. For the optimization, we use the Matlab toolbox YALMIP~\cite{Yalmip}, together with the solver SeDuMi~\cite{Sedumi}. We validate the optimization results with nonlinear simulation in the commercial power system simulation software Simscape Power Systems and PSS$^\circledR$Sincal, to obtain a practically relevant evaluation.



\subsection{The IEEE 39 bus 10 generator model}
\label{subsec.IEEE39}

The first example is a dynamic model of the IEEE 39 bus, and 10 power plant system, which is adopted from~\cite{moeini2015open}. The topology of the power system is shown in Fig.~\ref{fig.IEEE39}. It consists of 10 power plants whose structure is described in detail in~\ref{App.IEEE39Models}. The power system contains static prosumers, denoted with arrows, c.f. Fig.~\ref{fig.IEEE39}. We consider the active powers of constant-power elements in buses 2, 4, 9, 21, 23, 26, 29 as disturbance inputs, marked with blue in Fig.~\ref{fig.IEEE39}.
All prosumer and grid parameters are taken from~\cite{moeini2015open}. We increased the exciter gains from 200 to 800 to obtain a stable system. The tunable controller parameters of all power plant controllers are marked red in Figs.~\ref{fig.Exciter}, \ref{fig.PSS}, and~\ref{fig.TGOV}. The overall linear system consists of 190 states and 100 tunable controller parameters, c.f.~\cite{Mesanovic18ACC,Mesanovic17ISGT}. {The initial maximal allowed step size for each controller parameter of each generator is shown in Table~\ref{tab.DKIEEE39} in~\ref{App.IEEE39Models}.}

\begin{figure}[t]
	\centering
	\includegraphics[width=1.0\columnwidth]{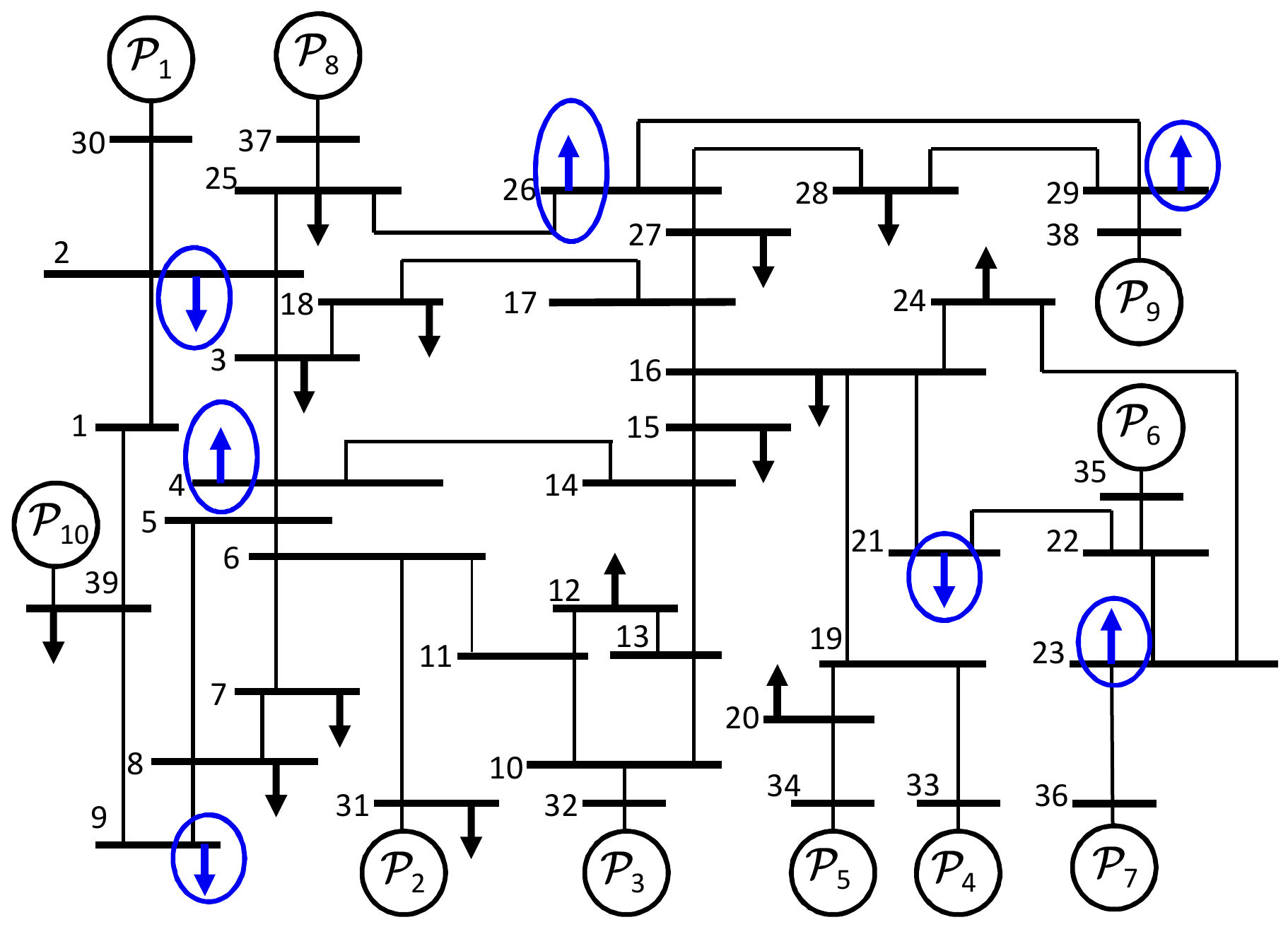}
	\caption{A IEEE 39 bus system with 10 dynamic power plant prosumers~\cite{Mesanovic18ACC,Mesanovic17ISGT}. Blue arrows denote the disturbances $w_i$.}
	\label{fig.IEEE39}
\end{figure}

\begin{figure}[tb]
	\centering
	\includegraphics[width=1.0\columnwidth]{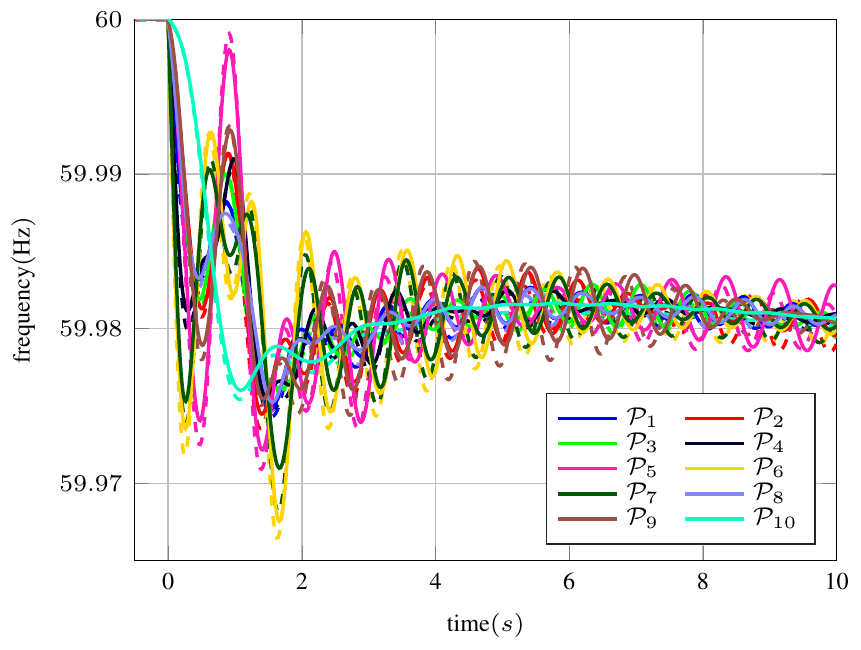}
	\caption{Initial parameters: frequency response after a 100 MW load step in bus 21. Solid lines represent simulations with the nonlinear model, whereas dashed lines represent simulations with the linear model.}
	\label{fig.IEEE39initF}
\end{figure}

\begin{figure}[tb]
	\centering
	\includegraphics[width=1.0\columnwidth]{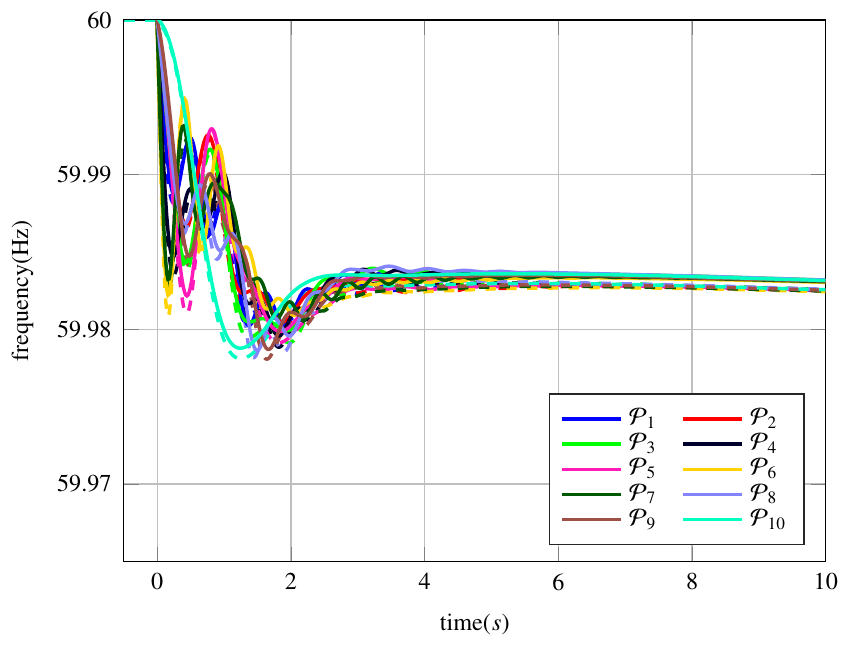}
	\caption{$\Hinf$ tuned parameters: frequency response after a 100 MW load step in bus 21. Solid lines represent simulations with the nonlinear model, whereas dashed lines represent simulations with the linear model.}
	\label{fig.IEEE39optF}
\end{figure}
\begin{figure}[tb]
	\centering
	\includegraphics[width=1.0\columnwidth]{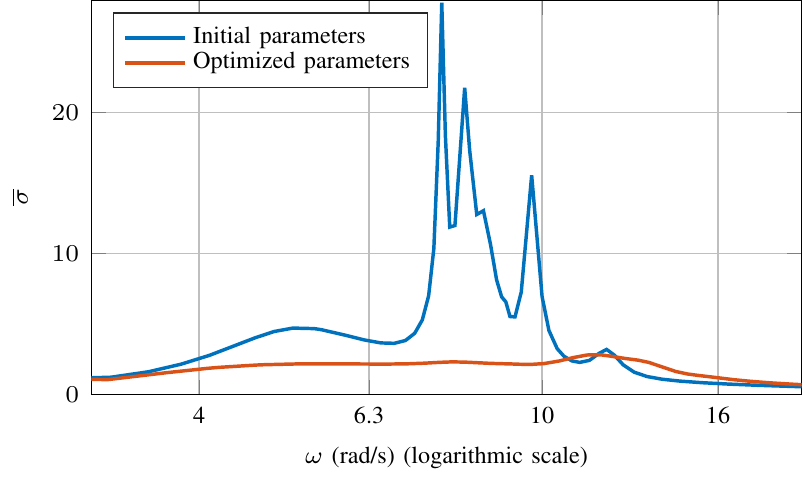}
	\caption{Largest singular value of the linearized IEEE 39 bus power system as a function of frequency $\omega$. After optimization, most of the resonant peaks in the system are eliminated.}
	\label{fig.IEEE39Sigmas}
\end{figure}

Figure~\ref{fig.IEEE39initF} shows the linear (dashed lines) and nonlinear (solid lines) simulation of the generator frequencies. The nonlinear simulation is performed in Simscape Power Systems~\cite{SimPowSys} with nonlinear models of the power plants and the nonlinear power flow. It shows poorly dampened oscillations in the system. Thereby, $\cP_{10}$ emulates a connected power system, and thus has a much larger inertia than other power plants. Consequently, its behavior in the time response in Fig.~\ref{fig.IEEE39initF} is different than the response of the other power plants.
The difference between the linear and nonlinear responses in Fig.~\ref{fig.IEEE39initF} is small and the linear model can be utilized for the optimization.

Figure~\ref{fig.IEEE39optF} shows the time-domain response using the proposed tuning algorithm, which is significantly improved.
Simulations with the optimized parameters of the linear model (dashed lines) again shows good correspondence to the detailed nonlinear simulation (solid lines). The structured controller synthesis reduced the $\Hinf$ norm by a factor of 10. Thus, the optimally tuned parameters reject disturbances significantly better than in the untuned case. The largest singular value of the system, presented in Fig.~\ref{fig.IEEE39Sigmas} as a function of the input frequency, shows that the resonant peaks were practically eliminated after the parameter optimization.

{
To further demonstrate the efficiency of the proposed approach, we optimize controller parameters considering independent (one-by-one) failures of all dynamic prosumers in the system, as well as the independent failures of power lines (PLs) between buses 2 and 3, 5 and 6, 6 and 11, 21 and 22, 23 and 24, as well as 28 and 29, denoted with PL 1-6, respectively. Such failures excite simultaneous changes of active- and reactive power. Furthermore, as failures occur almost instantaneously in the system, one controller parameterization is needed for all considered failures as there is not enough time to reparameterize the system after a failure.
} 
{
For this purpose, we apply the approach from~\cite{mesanovic2018ISGT} to convert such failures into equivalent static prosumers. The obtained model, however, is only valid for one considered failure. Hence, we obtain 16 independent systems, for which we find one controller parameterization by using~\eqref{eq.FiniteOptProbMultiScenario}.
Figure~\ref{fig.IEEE39_f_allgens_init} shows the results of 10 simulations using the nonlinear model for the initial parameters.
Each shaded area shows the frequency response of the system after the dropout of one prosumer. Thereby, the bounds of the frequencies of the remaining 9 dynamic prosumers is shown. 
It shows prevailing oscillations in the system after the failure of almost any prosumer. Figure~\ref{fig.IEEE39_f_allgens_opt} shows 10 simulations after the independent dropout of each dynamic prosumer with the nonlinear model for optimized parameters, demonstrating a significantly improved response, with a reduced overshoot and settling time in all cases.
}
{
Analogously, Fig.~\ref{fig.IEEE39_f_PLfailure_init} shows the result of six nonlinear simulations with initial parameters. Thereby, each shaded area shows the response of all 10 dynamic prosumers to the considered PL failure. Using the proposed approach, the response of the system to such failures is significantly improved, c.f. Fig.~\ref{fig.IEEE39_f_PLfailure_opt}.
}

\begin{figure}[tb]
	\centering
	\includegraphics[width=1.0\columnwidth]{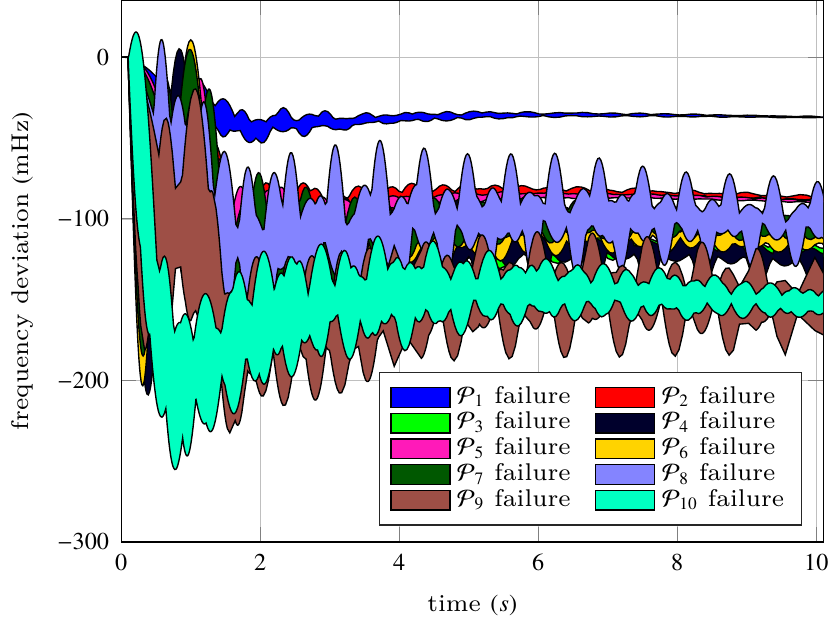}
	\caption{{Bounds of nonlinear frequency responses after dynamic prosumer failures with the initial parameterization. Each area shows the response of the remaining 9 dynamic prosumers. The figure shows ten independent nonlinear simulations.}}
	\label{fig.IEEE39_f_allgens_init}
\end{figure}

\begin{figure}[tb]
	\centering
	\includegraphics[width=1.0\columnwidth]{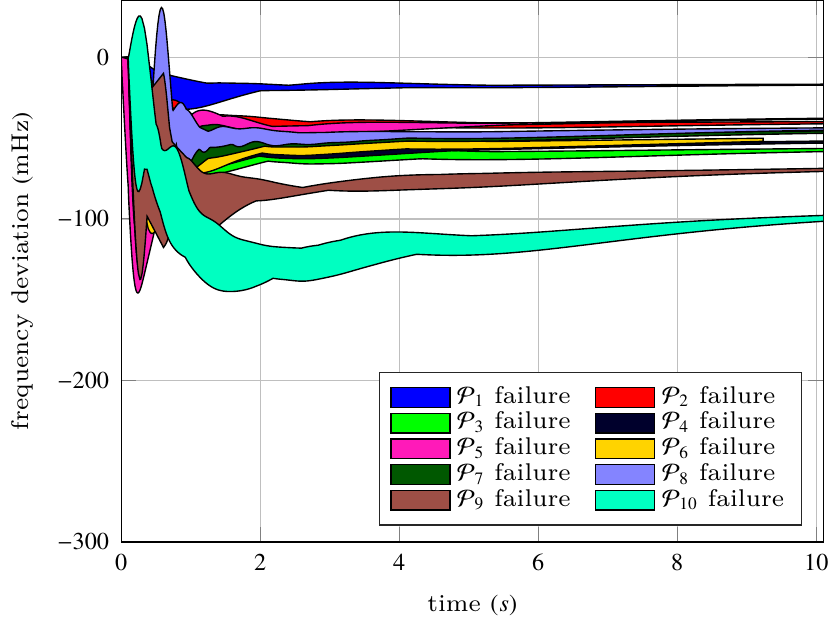}
	\caption{{Bounds of nonlinear frequency responses after dynamic prosumer failures with the optimized parameterization. Each area shows the response of the remaining 9 dynamic prosumers. The figure shows ten independent nonlinear simulations.}}
	\label{fig.IEEE39_f_allgens_opt}
\end{figure}

\begin{figure}[tb]
	\centering
	\includegraphics[width=1.0\columnwidth]{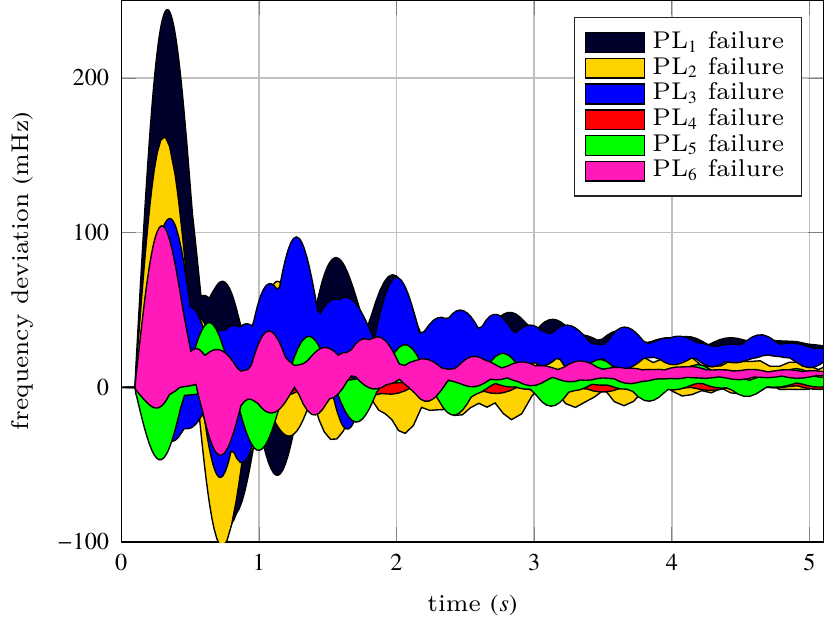}
	\caption{{Bounds of frequency responses after power line failures with the initial parameterization. Each area shows the response of all 10 dynamic prosumers. The figure shows six independent nonlinear simulations.}}
	\label{fig.IEEE39_f_PLfailure_init}
\end{figure}

\begin{figure}[tb]
	\centering
	\includegraphics[width=1.0\columnwidth]{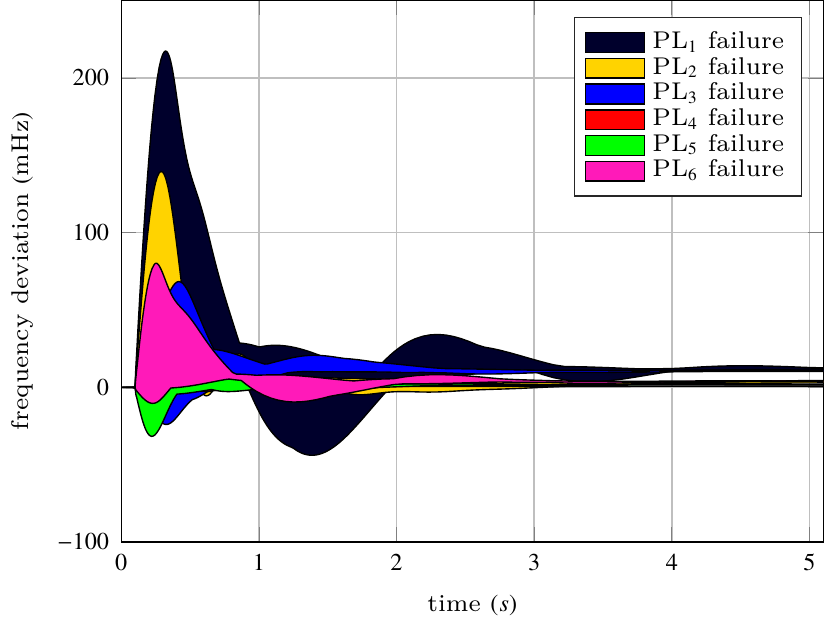}
	\caption{{Bounds of frequency responses after power line failures with the initial parameterization. Each area shows the response of all 10 dynamic prosumers. The figure shows six independent nonlinear simulations.}}
	\label{fig.IEEE39_f_PLfailure_opt}
\end{figure}

\subsection{European 53 generator model}
\label{subsec.DynaGrid}
\begin{figure}[t]
	\centering
	\includegraphics[width=1\columnwidth]{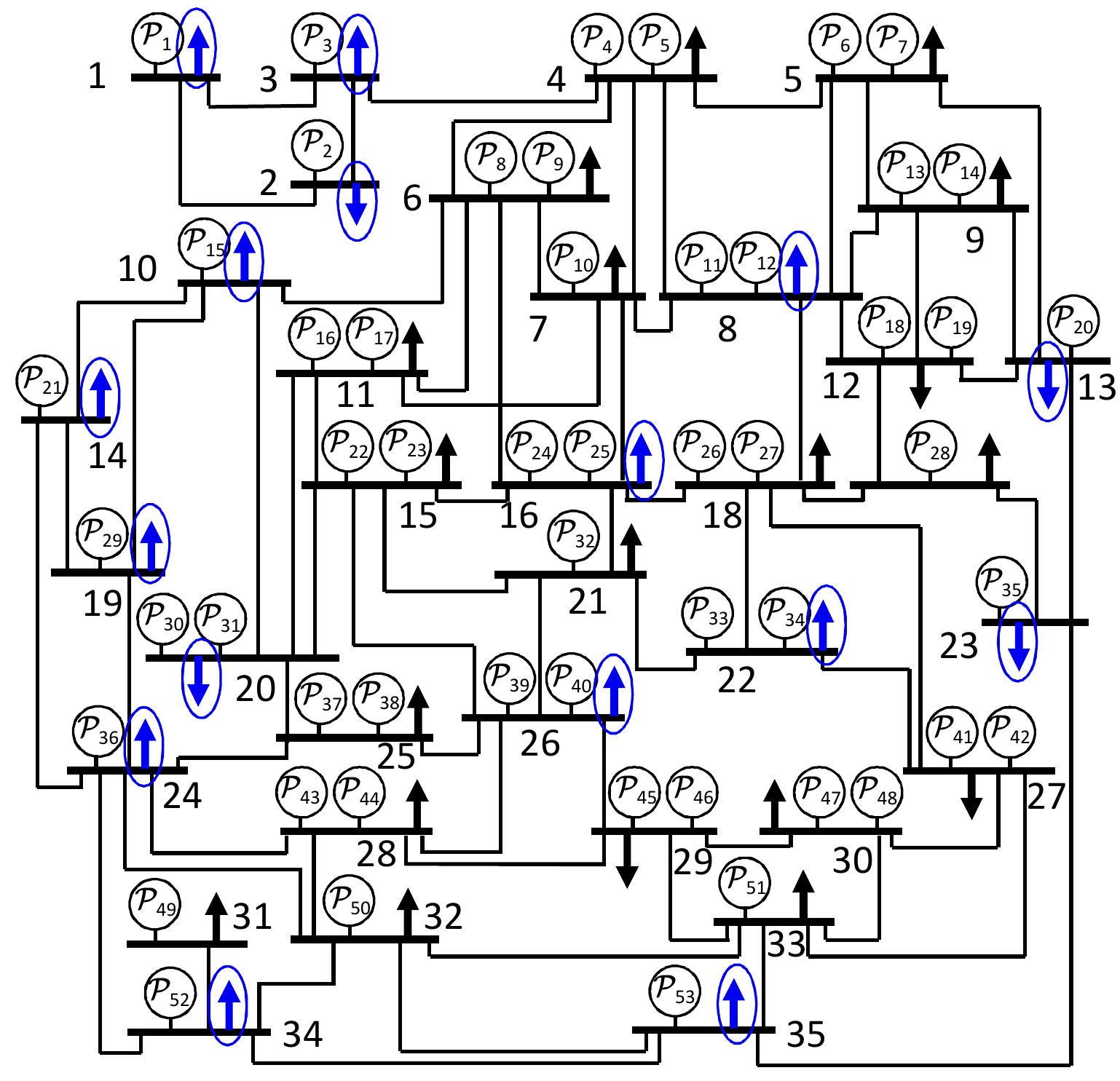}
	\caption{The 53 power plant power system developed in the DynaGridCenter project~\cite{dynagrid}. 
		Buses with uncertain infeeds, denoting disturbances $w_i$, are marked blue.}
	\label{fig.DynaGrid}
\end{figure}

For the second example, we use a model with 53 power plants. It represents a reduced version of the European power system, developed as a part of the research project DynaGridCenter~\cite{dynagrid}. An overview of the power system structure is shown in Fig.~\ref{fig.DynaGrid}. The grid consists of 35 buses (nodes), connected by long power lines. 
The controllers used for this model are presented in~\ref{App.DynaGridModels}. 
A more detailed description of the considered system is avoided as it is not necessary for the understanding of the presented results.
Nineteen power plants in the system have controllers, whereas all other power plants have a constant exciter voltage $E_{fd,i}$ and turbine mechanical power $P_{m,i}$. The described power system has a total of 469 states and 116 controller parameters. We consider the active powers of static prosumers in 15 buses as disturbance inputs, marked with blue in Fig.~\ref{fig.DynaGrid}.

\begin{figure}[tb]
	\centering
	\includegraphics[width=1.0\columnwidth]{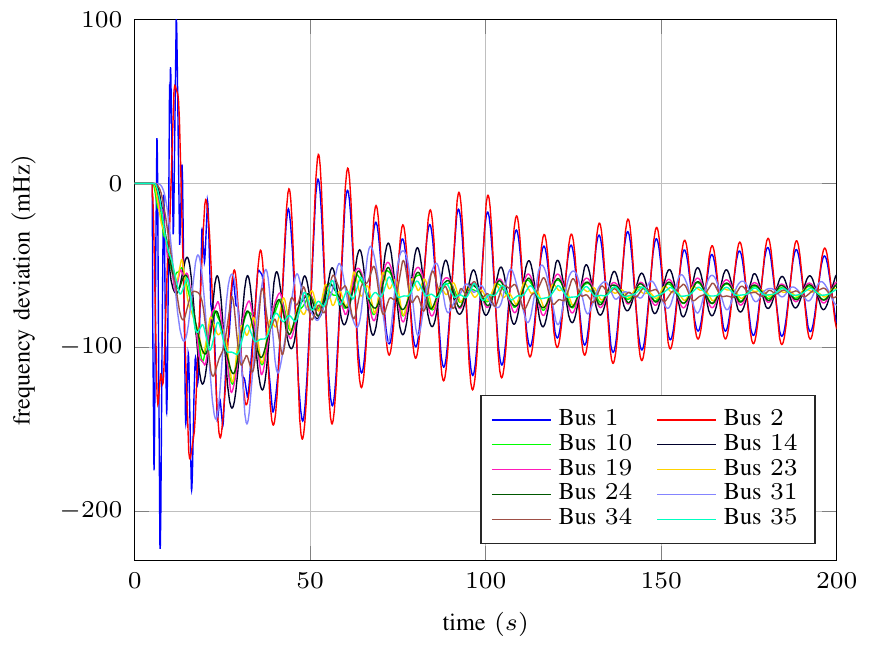}
	\caption{Frequency response in several buses in the system after a 1.5 GW generator outage in bus 1 with initial parameters.}
	\label{fig.DynaGridInitFreq}
\end{figure}
\begin{figure}[tb]
	\centering
	\includegraphics[width=1.0\columnwidth]{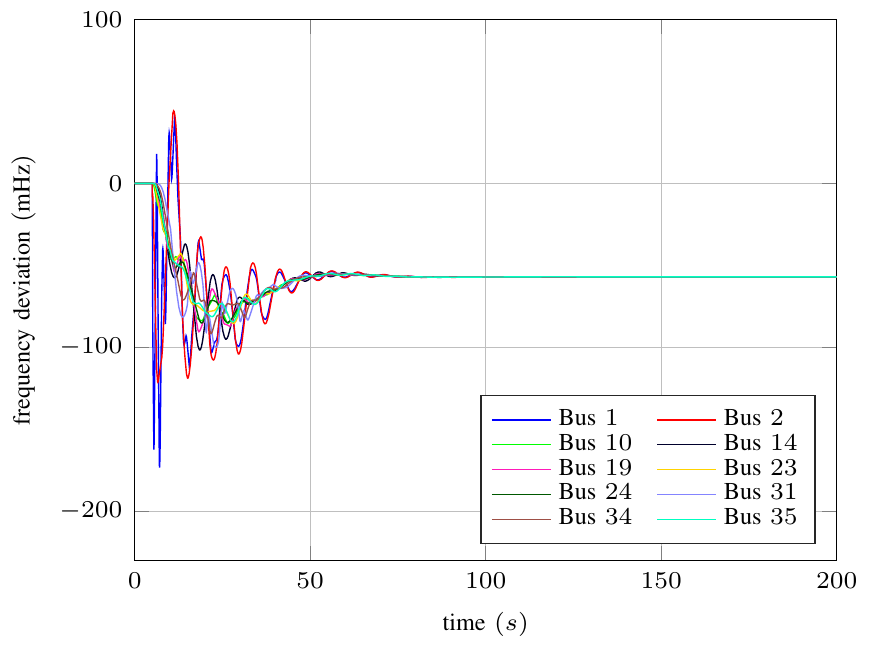}
	\caption{Power plant frequency response after a 1.5 GW load step in bus 1 with $\Hinf$ tuned parameters.}
	\label{fig.DynaGridOptFreq}
\end{figure}

Figure~\ref{fig.DynaGridInitFreq} shows the frequency response at ten nodes in the system with initial system parameters after a 1.5 GW generation dropout in bus 1. The simulation is done with nonlinear power plant and power grid models in the commercial power system simulation software PSS{$^\circledR$}Sincal. As shown in Fig.~\ref{fig.DynaGridInitFreq}, poorly dampened oscillations are present in the system. The initial parameters were obtained manually with iterative simulation. Due to the system complexity and time limitations, we did not find a better parameterization manually. The proposed tuning algorithm provides a systematic way to tune the parameters of such complex systems.

Figure~\ref{fig.DynaGridOptFreq} shows the frequency response after the application of the $\Hinf$ tuning algorithm. It shows a reduction in the overshoot, as well as significantly improved oscillation damping, confirmed by the singular value plot in Fig.~\ref{fig.DynaGridSigmas}. The system $\Hinf$ norm was reduced by a factor of 5.4, and thereby most of the resonant peaks were practically eliminated; c.f. Fig.~\ref{fig.DynaGridSigmas}, even though the parameters of only 19 power plants, from a total of 53, were optimized.

\begin{figure}[tb]
	\centering
	\includegraphics[width=1.0\columnwidth]{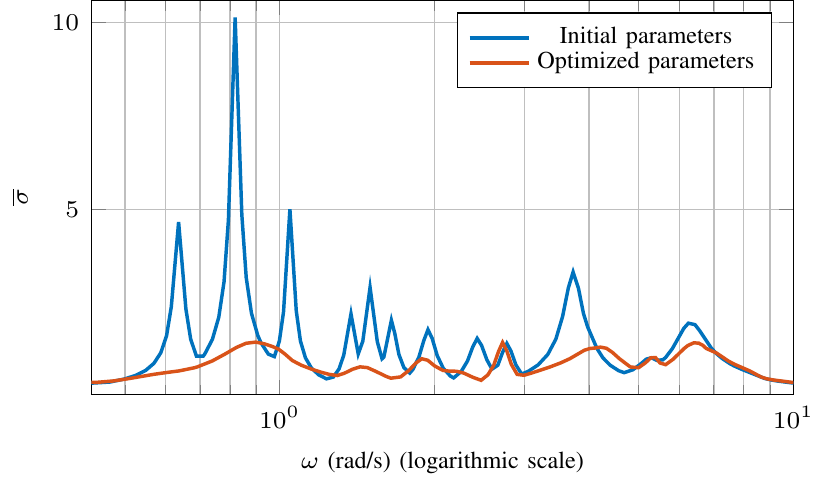}
	\caption{Largest singular value of the linearized reduced European power system model as a function of frequency $\omega$. After optimization, most of the resonant peaks in the system are practically eliminated.}
	\label{fig.DynaGridSigmas}
\end{figure}

\subsection{Discussion}
\label{subsec.Discussion}
We performed simulation studies on two power system models exploiting linearized models for the tuning and using nonlinear simulation environments for verification. The $\Hinf$ tuning algorithm reduced the $\Hinf$ norm of the systems to 0.01\% and 19\% of the initial norm with static prosumers as disturbances, thereby significantly reducing the time-domain settling time and overshoot of those systems, as summarized in Table~\ref{tab.NumEvalSummary}. The presented approach provides a systematic solution and shows very good results for parameter tuning in these complex systems. The outcomes of the optimization are also validated in commercial power system simulation software with detailed nonlinear component models, showing the applicability of the approach to practical systems.



\begin{table}[tb]
	\centering
	\caption{Comparison of the $\Hinf$ norm, settling time, and overshoot improvement for the two power system simulation studies with static prosumers as disturbances.}
	\label{tab.NumEvalSummary}
	\begin{tabular}{c c c c}
		\toprule       
		relation $\frac{x_\text{opt}}{x_\text{initial}}$ & $\Hinf$ norm & Settling time & Overshoot \\ \midrule
		IEEE 39 bus                                  & 0.1         & 0.72          & 0.66      \\
		European                                   & 0.19         & 0.11          & 0.78      \\ \bottomrule
	\end{tabular}
\end{table}

\section{Experimental validation}
\label{sec.ExpEval}
\begin{figure}[tb]
	\centering
	\includegraphics[width=1\columnwidth]{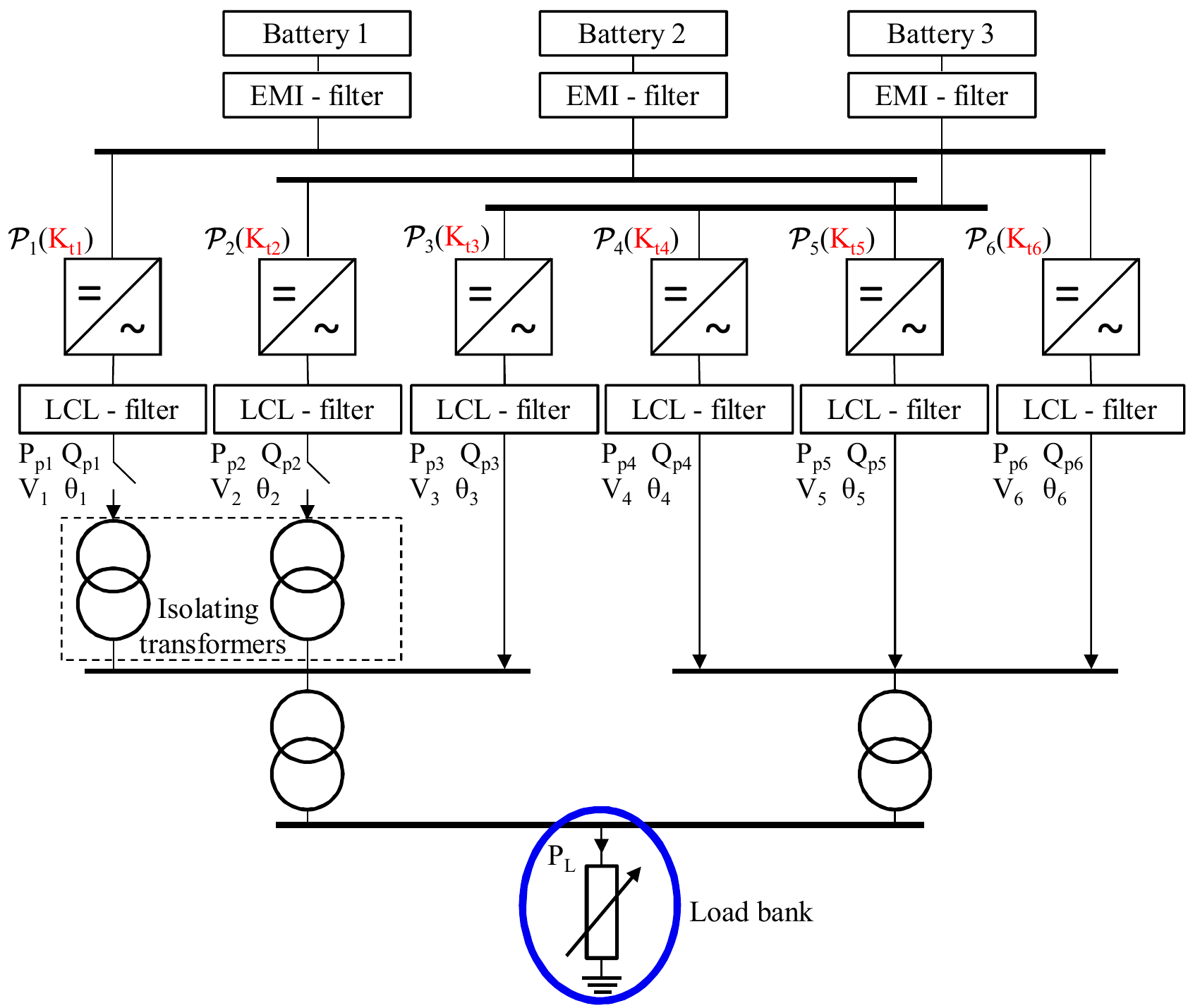}
	\caption{Structure of the considered part of the microgrid, consisting of six parallel connected inverters. Details can be found in~\cite{rahmoun2017mathematical}. The active power of the load bank, denoted with $P_L$, is the disturbance input $w_i$ into the system.}
	\label{fig.IREN2Structure}
\end{figure}
The $\Hinf$ tuning algorithm was furthermore validated on a testbed microgrid in Wildpoldsried, Germany, as a part of a funded research project~\cite{iren2}.

The considered part of the grid consists of six 55 kVA SINAMICS inverters, connected to three Lithium-Ion batteries, and a controllable 150 kW load bank, c.f. Fig.~\ref{fig.IREN2Structure}. The microgrid can operate attached to the supply grid, as well as in islanded operation. Further details can be found in~\cite{iren2,rahmoun2017mathematical}.
We consider the case when the microgrid is running independently of the supply grid.
All inverters are running in grid-forming mode, i.e. they control their voltage magnitude and frequency based on their active and reactive power infeed, as shown in Fig.~\ref{fig.InvModel}. This leads to increased reliability and power quality in the system, because failure of one inverter will not cause a blackout when properly configured.
In order to enable parallel operation of the inverters, droop control of active and reactive power is used, as described in Subsection~\ref{subsec.ProsumerModel}. Droop control is the current state-of-the-art method for control of distributed generations for several reasons: it requires only local measurements and no real-time communication or accurate time synchronization, it enables power sharing and parallel operation of grid-forming inverters etc. {Other control schemes for inverters, beside droop control, are also possible, such as virtual oscillator control~\cite{gross2019effect} or the virtual synchronous generator concept~\cite{chen2019parameter}.}
We perform load steps with the load bank in order to evaluate the system performance.

The presented system is of interest for several reasons:
\begin{itemize}
	\item To avoid circulating currents, two isolating transformers are a part of the system, as shown in Fig.~\ref{fig.IREN2Structure}, which are sources of asymmetry in the load-step response of the inverters. Such asymmetry will also occur if the inverters are geographically distributed within a microgrid. Therefore this configuration is a good test example for a real life setup.
	\item 
	The system does not have any generation with mechanical inertia. Hence, it is an interesting example for a zero-inertia system.
\end{itemize}

%
%

\subsection{Manual tuning of the testbed system}
\label{Sec.IREN2ManualOptimization}

Manual tuning of the system was performed with iterative simulation methods based on the inverter model described in Subsection~\ref{subsec.ProsumerModel}, see Fig.~\ref{fig.InvModel}. The step response in Fig.~\ref{fig.2InvMeasSim} is obtained with the parameters from Table~\ref{tab.2InvInitStab}. It shows good correspondence between measurement (solid lines) and simulation (dashed lines), demonstrating the validity of the used model. 
The difference between measurement and simulation originates from unmodeled loads, other inverter controllers, phase-asymmetries etc.
We show active power plots, because in this system, the oscillations are better visible in the active power than in the frequencies.

\begin{table}[tb]
	\centering
	\caption{Stable manual parameterization of inverters 1 and 6.}
	\label{tab.2InvInitStab}
	\begin{tabular}{c c c c c}
		\toprule       Inv & $K_P$ (\%) & $K_Q$ (\%) & $T_f$ (ms) & $T_v$ (ms) \\ \midrule
		1               & 2         & 3.1       & 100        & 100        \\
		6               & 2         & 3.1       & 100        & 100        \\ \bottomrule
	\end{tabular}
\end{table}

\begin{figure}[tb]
	\centering
	\includegraphics[width=1\columnwidth]{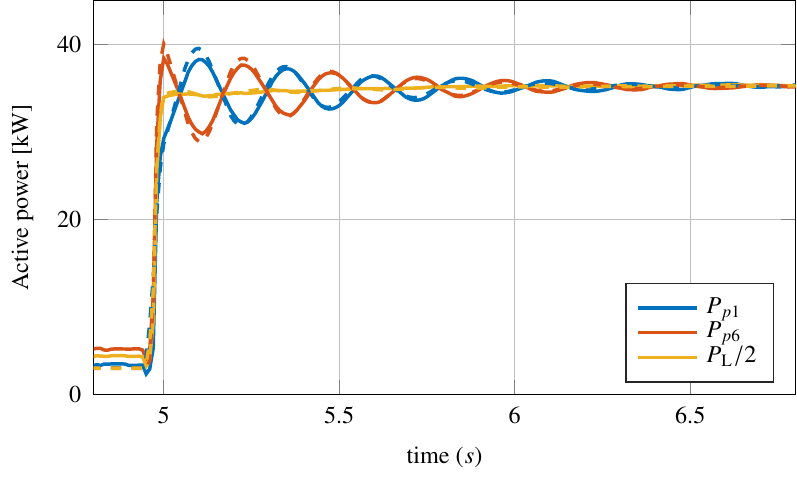}
	\caption{Response to a 60 kW load step with inverters 1 and 6 achieved by manual tuning; $P_{L}$ is calculated as the sum of $P_{p1}$ and $P_{p6}$.   Solid lines represent measurements, whereas dashed lines represent simulations with the nonlinear model.}
	\label{fig.2InvMeasSim}
\end{figure}

The same parameters from Table~\ref{tab.2InvInitStab} are also used for the operation of all six inverters, resulting in the 150 kW load step response shown in Fig.~\ref{fig.6InvStabSimMeas}. A discrepancy is present in the oscillation frequency between measurement and simulation of inverter 3. 
A better match can be obtained by an iterative adaptation of grid parameters, i.e. impedances in the grid. 
We avoid this because mismatches between measurements and simulation are expected in real systems, and as it allows to test the sensitivity of the proposed method to model discrepancies. 
The setup with six inverters was also used for successful operation with household consumers in islanded mode.

\begin{figure}[tb]
	\centering
	\includegraphics[width=1\columnwidth]{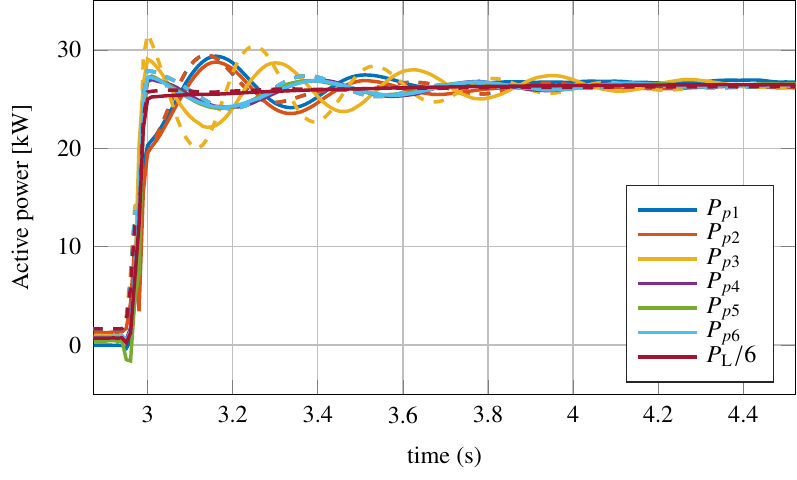}
	\caption{Response to a 150 kW load step with all inverters achieved by manual tuning; $P_{L}$ is calculated as the sum of all inverter powers. Solid lines represent measurements, whereas dashed lines represent simulations with the nonlinear model.}
	\label{fig.6InvStabSimMeas}
\end{figure}

\subsection{Automatic tuning of the testbed system}
\label{Sec.IREN2Optimization}

The results obtained by manual tuning in Figs.~\ref{fig.2InvMeasSim} and~\ref{fig.6InvStabSimMeas} show prevailing oscillations in the system after a load step. Arguably, they are still satisfactory for many applications. However, manual tuning requires expert know-how of the system and is associated with a significant time-effort.
Automatic tuning methods enable the fast design of robust microgrids, without expert knowledge. We apply and experimentally validate the proposed $\Hinf$ tuning method on the testbed system.

\subsubsection{Parameter tuning for inverters 1 and 6}

We first apply the $\Hinf$ parameter tuning algorithm to the system when only inverters 1 and 6 are running.
The response for a 60 kW load step with optimized parameters, c.f. Table~\ref{tab.2invOptAll}, is shown in Fig.~\ref{fig.2InvOptMeas}. The settling time 
of the step response is practically reduced to zero. 
{
Measurements with several load steps are shown in Fig.~\ref{fig.2InvOptAll_loadSteps}, demonstrating that the linear model is valid for a range of loading conditions and disturbances.
}
However, due to different droop values of $K_{P,1}$ and $K_{P,6}$, the steady state power of the inverters is not identical. Such parameterization may cause inverter 6 to overload after a large load step.
Still, if the inverters have sufficient power reserves, and no large sudden load changes are expected, this parameterization provides the best step response with regard to oscillation suppression. The generation of the inverters can be balanced out with slower control schemes, called secondary control, which are a standard part of power system control. As they operate at a slower time scale than the ones observed here, they are beyond the scope of this work.

\begin{figure}[tb]
	\centering
	 \includegraphics[width=1\columnwidth]{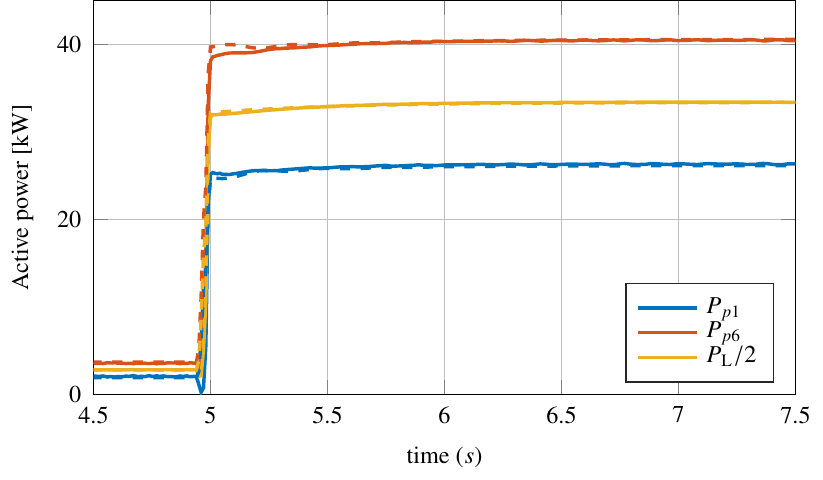}
	\caption{Response to a 60 kW load step with inverters 1 and 6 achieved by optimal tuning of all parameters; $P_{L}$ is calculated as the sum of $P_{p1}$ and $P_{p6}$. The optimized parameters are shown in Table~\ref{tab.2invOptAll}. Solid lines represent measurements, whereas dashed lines represent simulations with the nonlinear model.}
	\label{fig.2InvOptMeas}
\end{figure}
\begin{figure}[tb]
	\centering
	\includegraphics[width=1\columnwidth]{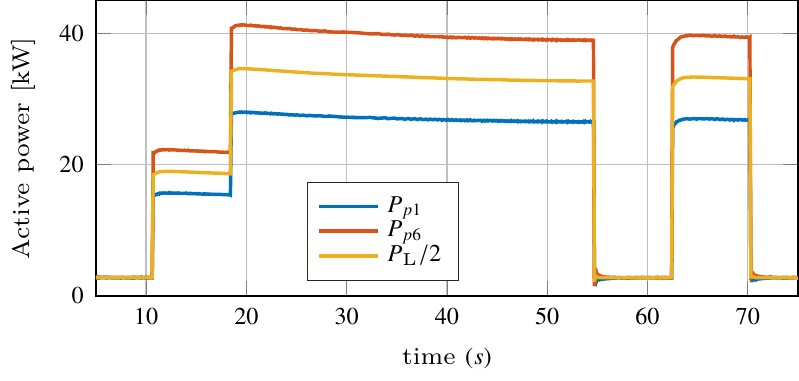}
	\caption{{Response of inverters 1 and 6 to 30 and 60 kW load steps with parameters from Table~\ref{tab.2invOptAll}.}}
	\label{fig.2InvOptAll_loadSteps}
\end{figure}

In order to eliminate the generation imbalance even without secondary control, we introduce additional constraints which enforce the equality of the droop gains, i.e. $K_{P1} = K_{P6}$ and $K_{Q1} = K_{Q6}$. With these constraints, we obtain optimized parameters shown in Table~\ref{tab.2invOptTime}, which achieve the step response shown in Fig.~\ref{fig.2InvOptTimeMeas}. We see that, even with the equality constraint, improvement in the step response of the system is still possible, compared to manual tuning results.


\begin{table}[tb]
	\centering
	\caption{Optimal parameterization of inverters 1 and 6 when all parameters are optimized.}
	\label{tab.2invOptAll}
	\begin{tabular}{c c c c c}
		\toprule       Inv & $K_P$ (\%) & $K_Q$ (\%) & $T_f$ (ms) & $T_v$ (ms) \\ \midrule
		1               & 3.13       & 3.56       & 108        & 104        \\
		6               & 2          & 3.62       & 115        & 104        \\ \bottomrule
	\end{tabular}
\end{table}

\begin{figure}[tb]
	\centering
	\includegraphics[width=1\columnwidth]{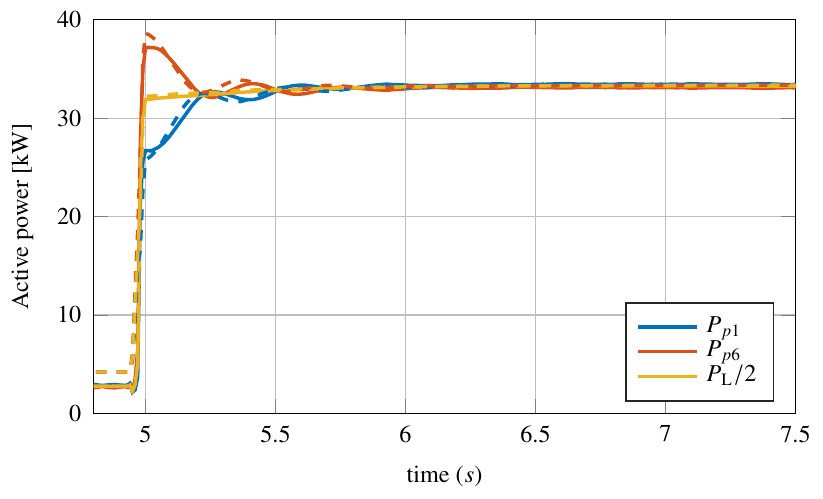}
	\caption{Response to a 60 kW load step with inverters 1 and 6 achieved by optimal tuning together with droop gain equality constraints; $P_{L}$ is calculated as the sum of $P_{p1}$ and $P_{p6}$. The optimized parameters are shown in Table~\ref{tab.2invOptTime}. Solid lines represent measurements, whereas dashed lines represent simulations with the nonlinear model.}
	\label{fig.2InvOptTimeMeas}
\end{figure}

\begin{table}[tb]
	\centering
	\caption{Optimal parameterization of inverters 1 and 6 with droop equality constraints.}
	\label{tab.2invOptTime}
	\begin{tabular}{c c c c c}
		\toprule       Inv & $K_P$ (\%) & $K_Q$ (\%) & $T_f$ (ms) & $T_v$ (ms) \\ \midrule
		1               & 2          & 3.13       & 89         & 100        \\
		6               & 2          & 3.13       & 130        & 100        \\ \bottomrule
	\end{tabular}
\end{table}

\subsubsection{Parameter tuning for all inverters}
\label{Subsec.6inv}

All 6 inverters are operating in parallel in grid-forming mode. 
The 150 kW load step response when all tunable inverter parameters are optimized, is shown in Fig.~\ref{fig.6InvOptAll}. The optimized parameters are shown in Table~\ref{tab.6invOptAll}. In this case, the oscillations could not be completely eliminated because of insufficient freedom in the parameterization. Still, a noticeable improvement is observable compared to manual tuning, c.f. Fig.~\ref{fig.6InvStabSimMeas}. 

To avoid unequal power sharing, equality constraints for the droop gains are introduced, i.e. $K_{P,1} = K_{P,2} = ... = K_{P,6}$ and $K_{Q,1} = K_{Q,2} = ... = K_{Q,6}$.
The step response for this case is shown in Fig.~\ref{fig.6InvOptTime}, and the obtained parameters in Table~\ref{tab.6invOptTime}. The overshoot in this case cannot be avoided. However, the power oscillations after the initial overshoot are reduced when compared to the manual tuning results in Fig.~\ref{fig.6InvStabSimMeas}.


\begin{figure}[tb]
	\centering
	\includegraphics[width=1\columnwidth]{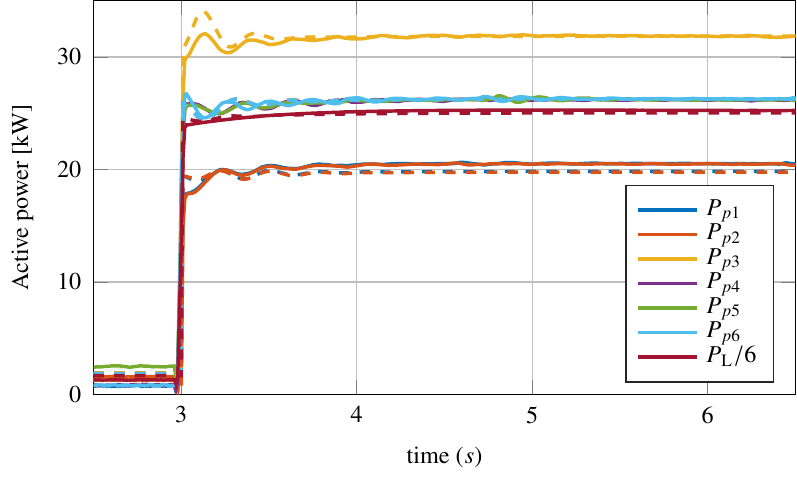}
	\caption{Response to a 150 kW load step with all inverters achieved by optimal tuning of all parameters; $P_{L}$ is calculated as the sum of all inverter powers. The optimized parameters are shown in Table~\ref{tab.6invOptAll}. Solid lines represent measurements, whereas dashed lines represent simulations with the nonlinear model.}
	\label{fig.6InvOptAll}
\end{figure}
\begin{table}[tb]
	\centering
	\caption{Optimal parameterization of inverters 1 and 6 with no droop equality constraints.}
	\label{tab.6invOptAll}
	\begin{tabular}{c c c c c}
		\toprule       Inv & $K_P$ (\%) & $K_Q$ (\%) & $T_f$ (ms) & $T_v$ (ms) \\ \midrule
		1-2               & 3.1        & 3.3        & 107        & 105        \\
		3               & 2          & 3.5        & 124        & 104        \\
		4-6               & 2.1        & 3.6        & 102        & 104        \\ \bottomrule
	\end{tabular}
\end{table}
\begin{figure}[tb]
	\centering
	\includegraphics[width=1\columnwidth]{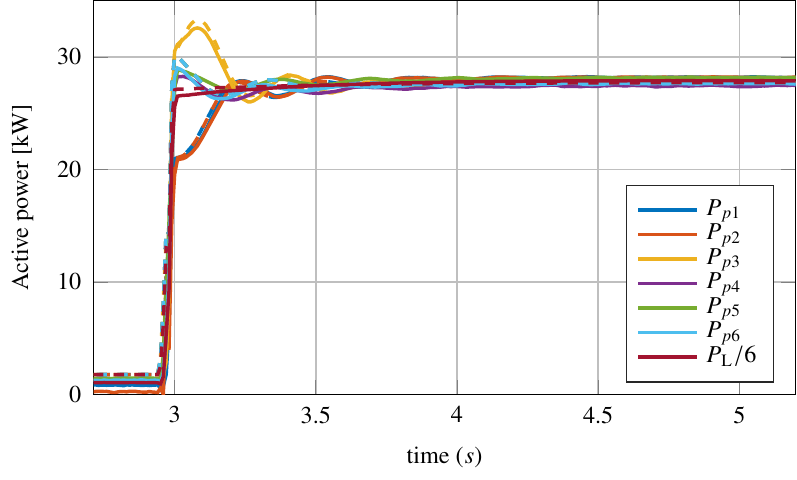}
	\caption{Response to a 150 kW load step with all inverters 1 and 6 achieved by optimal tuning and droop equality constraints; $P_{L}$ is calculated as the sum of all inverter powers. The optimized parameters are shown in Table~\ref{tab.6invOptTime}. Solid lines represent measurements, whereas dashed lines represent simulations with the nonlinear model.}
	\label{fig.6InvOptTime}
\end{figure}
\begin{table}[tb]
	\centering
	\caption{Optimal parameterization of inverters 1 and 6 with droop equality constraints.}
	\label{tab.6invOptTime}
	\begin{tabular}{c c c c c}
		\toprule       Inv & $K_P$ (\%) & $K_Q$ (\%) & $T_f$ (ms) & $T_v$ (ms) \\ \midrule
		             1-2               & 2          & 3.1        & 86         & 96         \\
		             3               & 2          & 3.1        & 154        & 101        \\
		            4-6              & 2          & 3.1        & 123        & 100        \\ \bottomrule
	\end{tabular}
\end{table}


Summarizing, the results show a good match between measurements and the inverter models. Even though the manual tuning results, as shown in Figs.~\ref{fig.2InvMeasSim} and~\ref{fig.6InvStabSimMeas}, are satisfactory for many applications, the results obtained with the proposed parameter tuning  algorithm, shown in Figs.~\ref{fig.2InvMeasSim}~-~\ref{fig.6InvOptTime}, are better with respect to oscillation damping. Additionally, no expert knowledge is necessary for the automatic parameterization, and the parameterization process can be done in less time than by manual tuning.

{
\section{Performance comparison}
\label{sec.PerfComparison}
}
{
In this section, we focus on the computational efficiency of methods considering structured $\Hinf$ controller synthesis. For this purpose, we adapt five methods from literature to the considered application. Thereby, we observe the computation time of these methods on several numerical examples. An additional comparison is presented in~\cite{mesanovic2019comparison}, where the focus of the comparison is on the deliverable results in the time- and frequency domain, observed on an exemplary numerical example. We use for the comparison a Windows computer with an Intel$^\circledR$ i7-4810MQ CPU running at 2.8 GHz and with 8 GB of RAM. Note that the presented times should only give a hint about the computational complexity, as more tailored methods would allow to further decrease the computation time.
}

{
The comparison is done with respect to the computation times, achieved $\Hinf$ norm, as well as scalability. For this purpose we consider the PK-iteration~\cite{Mesanovic18ACC}, path-following method~\cite{Hassibi1999}, linearized convex-concave decomposition~\cite{dinh2012combining}, non-smooth optimization from  the systune toolbox in MATLAB~\cite{systune}, and the projection approach from~\cite{kanev2004robust}.
Thereby, we expand these methods to be applicable to nonlinear parameter dependencies. Other methods, which assume a static output feedback formulation of the problem, such as cone complementarity linearization and sequential linear programming matrix method~\cite{schuler2011design,befekadu2006robust} are not included in the comparison, as they introduce transformations only applicable to specific linear parameter dependencies.
}

{
We begin the comparison with a small system with two power plants, 28 states and 16 optimization parameters. The system is obtained by taking the four power plant grid from~\cite{kundur93a}, disconnecting two power plants, and by dividing the load in half. The initial $\Hinf$ norm of the system is 23. 
The grid $\Omega = \left\{ \omega_k \Big| \omega_k = 4 + 0.1\cdot (k-1), k = 1...31 \right\}$ is chosen for the frequency sampling method. 
Table~\ref{tab.2Gen_HinfSynthesisComparison} summarizes the comparison results for this system. All methods improved the system $\Hinf$ norm.  Thereby, the systune toolbox, deploying non-smooth optimization, achieves the smallest $\Hinf$ norm. The frequency sampling method achieves a similar norm, but in 20\% of time systune required. They are followed by the PK iteration, which achieves a similar system norm as the convex-concave decomposition, but in 75\% of the time. 
The path-following method achieves a worse system norm, but with the second-fastest time. 
With the projection method, the obtained system norm is the largest, and, even on the small system, the optimization time is over 20 minutes. This is due to the necessity for an eigenvalue decomposition in every step in the inner optimization in one iteration. Since our goal is to find scalable optimization methods, we do not consider the projection method in larger examples.
}

{
\begin{table}[tb]
	\centering
	\caption{Comparison of structured $\Hinf$ synthesis methods on a system with two power plants, 28 states and 16 optimization parameters. The initial $\Hinf$ norm of the system is 23.}
	\label{tab.2Gen_HinfSynthesisComparison}
	\begin{tabular}{c c c c}
		\toprule       
		method      & $\Hinf$ norm & \makecell{comp. \\ time} & out. iter \\ \midrule
		\textbf{Frequency samp.}              & \textbf{0.48}         & \textbf{6 s}              & \textbf{3}                \\
		PK iter.                 & 0.72         & 83 s             & 50               \\
		Path-following                 & 1.1          & 29 s             & 16               \\
		\makecell{Convex-concave \\ decomposition}         & 0.83         & 112 s            & 50               \\
		\textit{systune}               & 0.42         & 31 s             & NA               \\
		Projection method              & 4.8          & 1300 s           & 7                \\ \bottomrule
	\end{tabular}
\end{table}
}

{
The second example is the four power plant system from~\cite{kundur93a}. This system has 56 states, 32 optimization parameters, and an initial $\Hinf$ norm of 11.5. Table~\ref{tab.4Gen_HinfSynthesisComparison} summarizes the results with the considered methods. Again, systune achieved the smallest $\Hinf$ norm of the system, reducing the norm to approx. 2.3\% of the initial value. The frequency sampling method, using the same frequency grid as in the previous example, achieved similar results by reducing the system norm to 4.9\% of the initial value, but with a 25 times faster computation time. The PK iteration achieves the third-best $\Hinf$ norm, with a computation time similar to \textit{systune}.
}

\begin{table}[tb]
	\centering
	\caption{Comparison of structured $\Hinf$ synthesis methods on a system with four power plants, 56 states and 32 optimization parameters. The initial $\Hinf$ norm of the system is 11.5.}
	\label{tab.4Gen_HinfSynthesisComparison}
	\begin{tabular}{c c c c}
		\toprule       
		method & $\Hinf$ norm & \makecell{comp. \\ time} & out. iter. \\ \midrule
		\textbf{Frequency samp.}        & \textbf{0.56}         & \textbf{27 s}             & \textbf{5}                \\
		PK iter.           & 0.65         & 670 s            & 50               \\
		Path-following          & 1.9          & 640 s            & 50               \\
		\makecell{Convex-concave \\ decomposition}   & 1.18         & 1264 s           & 50               \\
		\textit{systune}         & 0.27         & 687 s            & NA               \\ \bottomrule
	\end{tabular}
\end{table}

{
The third considered example is the ten power plant system from~\ref{App.IEEE39Models}, consisting of 190 states and 100 controller parameters. For this system, most of the methods reach the limit for practically tolerable computation times. The tuning with \textit{systune} could not be done due to an ''out of memory`` error. The optimization results are shown in Table~\ref{tab.10Gen_HinfSynthesisComparison}. The PK iteration, path-finding, and convex-concave decomposition have large computation times due to the presence of the Lyapunov matrix, whose size scales quadratically with the number of states. For this system, the total computation time for these methods is in the range of one to several days.
Only the frequency sampling method was able to find a solution in reasonable time by using the grid $\Omega = \{0.01, 3, 4, 5, 6 \} \bigcup \left\{ \omega_k \Big| \omega_k = 7 + 0.1\cdot (k-1), k = 1...81 \right\}$. 
Thereby, the density of the grid was increased in the frequency interval with resonant peaks, shown in Fig.~\ref{fig.IEEE39Sigmas}.
}

\begin{table}[tb]
	\centering
	\caption{Comparison of structured $\Hinf$ synthesis methods on the IEEE 39 bus system with ten power plants, 190 states and 100 optimization parameters. The initial $\Hinf$ norm of the system is 27.7.}
	\label{tab.10Gen_HinfSynthesisComparison}
	\begin{tabular}{c c c c}
		\toprule       
		method     & $\Hinf$ norm & comp. time                     & out. iter.       \\ \midrule
		             \textbf{Frequency samp. }              &\textbf{ 2.33}         & \textbf{266s}                          & \textbf{4}                \\
		                 PK iter.                  &              & \makecell{3800s\\ per iter.}   &                  \\
		              Path-following               &              & \makecell{8257s \\ per iter.}  &                  \\
		\makecell{Convex-concave \\ decomposition} &              & \makecell{13700s \\ per iter.} &                  \\
		             \textit{systune}              &              & \makecell{out of \\ memory}    &                  \\ \bottomrule
	\end{tabular}
\end{table}


\begin{figure}[tb]
	\centering
	\subfigure[$\Hinf$ norm.]{%
		\includegraphics[width=0.5\textwidth]{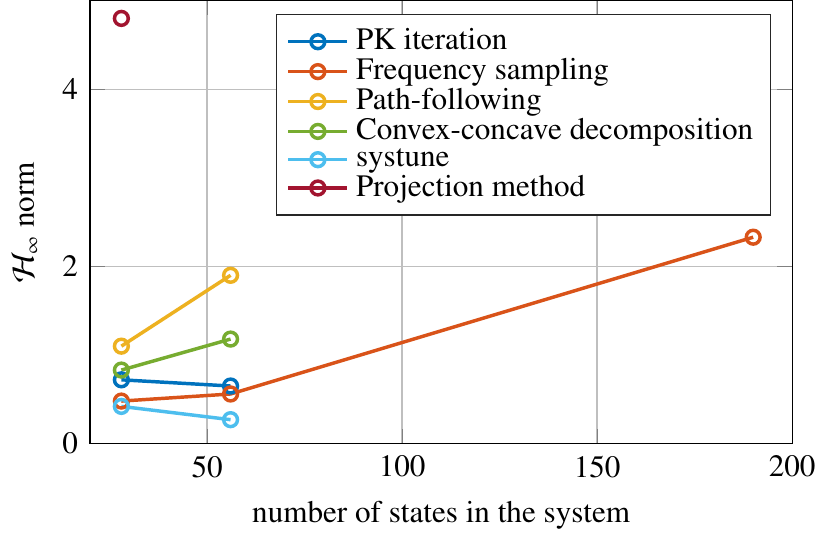}
		\label{fig.HinfNormComparison}}
	\subfigure[Computation times. Squares represent estimated times.]{%
		\includegraphics[width=0.5\textwidth]{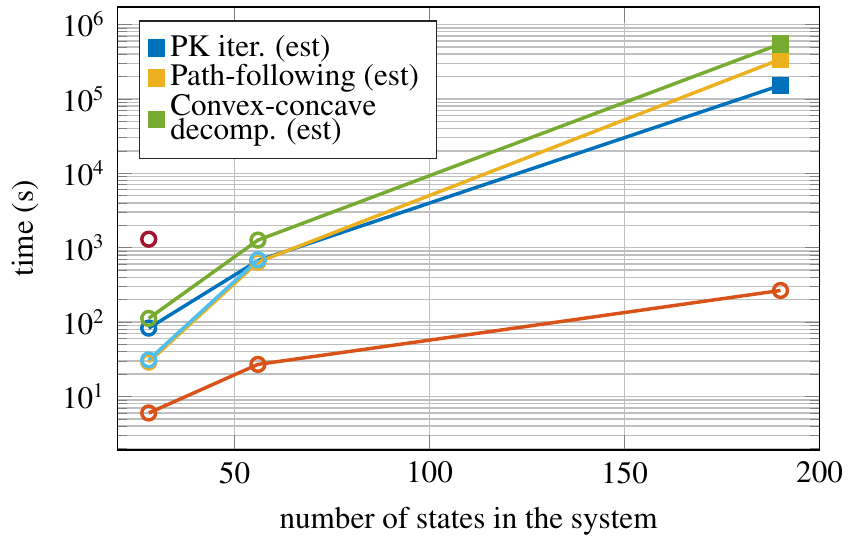}
		\label{fig.timeComparison}}
	\caption{Comparison of computation times and achieved $\Hinf$ norms for the different methods.}
	\label{fig.ComparisonAll}
\end{figure}

{
Figure~\ref{fig.HinfNormComparison} summarizes the obtained $\Hinf$ norms from the previous tables. It shows that the proposed frequency sampling method is only partially outperformed by \textit{systune} for the 4-power plant system. However, it is the only one capable of tuning the parameters of the 10-power plant system in reasonable time.
A summary of the computation times for the different methods is shown in Fig.~\ref{fig.timeComparison}. Thereby, squares in the plot represent estimated times obtained by multiplying the computation time per iteration from Table~\ref{tab.10Gen_HinfSynthesisComparison} with the number of iterations from the smaller systems. It shows that the frequency sampling method achieves orders-of-magnitude smaller computation times, while being only behind systune with respect to the achieved $\Hinf$ norm. However, the performance of the frequency sampling is dependent on the size of the frequency grid. It is important to choose the smallest grid that covers the necessary frequency area with sufficient density. This, however, was not a problem for the considered systems.
}

\section{Conclusions}
\label{sec.Conclusion}

Tuning of existing controller parameters to reject disturbances in power systems, resulting in oscillations, allows to accommodate changing power system dynamics, e.g. due to an increasing share of renewable generation. 
Automatic tuning algorithms could allow the system operator to retune the parameters of the existing controllers to account for changes and disturbances.
We proposed an algorithm for structured $\Hinf$ controller synthesis and applied it in simulations and experiments to power systems. We proved that the proposed algorithm will produce stabilizing controller parameters given an initial stabilizing controller. 
We applied the $\Hinf$ optimization method in two simulation studies containing power systems with 10 and 53 generators. In both cases, the $\Hinf$ norm of the systems was reduced by more than a factor of five, while the time-response to disturbance steps was also improved. Furthermore, we experimentally evaluated the approach on a testbed islanded microgrid. As shown, the used inverter model corresponds well to measurements. Furthermore, the developed tuning method leads to much better results than one achieves by existing manual tuning, with less time and a reduced amount of necessary expert knowledge of the system. {Finally, we compared the proposed approach to others from the literature to demonstrate the scalability of the proposed approach. Future work will focus on more detailed inverter models, as well as including power grid dynamics in the power grid model.}

\bibliography{references_amer_mesanovic}{}

\begin{thebibliography}{10}
\expandafter\ifx\csname url\endcsname\relax
  \def\url#1{\texttt{#1}}\fi
\expandafter\ifx\csname urlprefix\endcsname\relax\def\urlprefix{URL }\fi
\expandafter\ifx\csname href\endcsname\relax
  \def\href#1#2{#2} \def\path#1{#1}\fi

\bibitem{kundur93a}
P.~Kundur, Power System Stability and Control, Mc{G}raw-Hill, 1993.

\bibitem{ENTSOEReport}
entsoe, Analysis of {CE} inter-area oscillations of 1st december 2016,
  \url{https://docstore.entsoe.eu/Documents/SOC%20documents/Regional_Groups_Continental_Europe/2017/CE_inter-area_oscillations_Dec_1st_2016_PUBLIC_V7.pdf}
  (December 2020).

\bibitem{Ren21Report2018}
REN21, Renewables 2018 global status report,
  \url{http://www.ren21.net/gsr-2018/} (Jan 2020).

\bibitem{AlAli14}
S.~Al~Ali, T.~Haase, I.~Nassar, H.~Weber, Impact of increasing wind power
  generation on the north-south inter-area oscillation mode in the {European}
  {ENTSO-E} system, {IFAC} Proceedings Volumes 47~(3) (2014) 7653--7658.

\bibitem{crivellaro2019beyond}
A.~Crivellaro, A.~Tayyebi, C.~Gavriluta, D.~Gro{\ss}, A.~Anta, F.~Kupzog,
  F.~D{\"o}rfler, Beyond low-inertia systems: Massive integration of
  grid-forming power converters in transmission grids, arXiv preprint
  arXiv:1911.02870 (2019).

\bibitem{markovic2019understanding}
O.~Markovic, U.and~Stanojev, E.~Vrettos, P.~Aristidou, G.~Hug, Understanding
  stability of low-inertia systems\url{engrxiv.org/jwzrq} (Feb 2019).

\bibitem{mesanovic2019hierarchical}
A.~Me{\v s}anovi{\' c}, U.~M{\" u}nz, R.~Findeisen, Scalable and data privacy
  conserving controller tuning for large-scale power networks, submitted, arXiv
  preprint arXiv:1911.01499 (2019).

\bibitem{Raoufat16}
M.~Raoufat, K.~Tomsovic, S.~Djouadi, Virtual actuators for wide-area damping
  control of power systems, IEEE Trans. Power Systems 31~(6) (2016) 4703--4711.

\bibitem{Pipelzadeh17}
Y.~Pipelzadeh, N.~Chaudhuri, B.~Chaudhuri, T.~Green, Coordinated control of
  offshore wind farm and onshore {HVDC} converter for effective power
  oscillation damping, IEEE Trans. Power Systems 32~(3) (2017) 1860--1872.

\bibitem{Zhu03}
C.~Zhu, M.~Khammash, V.~Vittal, W.~Qiu, Robust power system stabilizer design
  using $\mathcal{H}_\infty$ loop shaping approach, IEEE Trans. Power Systems
  18~(2) (2003) 810 -- 818.

\bibitem{befekadu2005robust}
G.~Befekadu, I.~Erlich, Robust decentralized structure-constrained controller
  design for power systems: an {LMI} approach, in: Power Systems Computation
  Conference, 2005.

\bibitem{MahmoudiNAPS15}
M.~Mahmoudi, J.~Dong, K.~Tomsovic, S.~Djouadi, Application of distributed
  control to mitigate disturbance propagations in large power networks, in:
  North American Power Symposium (NAPS), IEEE, 2015.

\bibitem{Preece13}
R.~Preece, J.~Milanovic, A.~M. Almutairi, O.~Marjanovic, Damping of inter-area
  oscillations in mixed {AC/DC} networks using {WAMS} based supplementary
  controller, IEEE Trans. Power Systems 28~(2) (2013) 1160--1169.

\bibitem{wu2016input}
X.~Wu, F.~D{\"o}rfler, M.~R. Jovanovi{\'c}, Input-output analysis and
  decentralized optimal control of inter-area oscillations in power systems,
  IEEE Trans. Pow. Sys. 31~(3) (2016) 2434--2444.

\bibitem{Schuler14}
S.~Schuler, U.~M{\"u}nz, F.~Allg{\"o}wer, Decentralized state feedback control
  for interconnected systems with application to power systems, Journal of
  Process Control 24~(2) (2014) 379--388.

\bibitem{Marinescu09}
B.~Marinescu, B.~Mallem, H.~Bourles, L.~Rouco, Robust coordinated tuning of
  parameters of standard power system stabilizers for local and global grid
  objectives, in: PowerTech, Bucharest, IEEE, 2009.

\bibitem{Rouco01}
L.~Rouco, Coordinated design of multiple controllers for damping power system
  oscillations, International Journal of Electrical Power \& Energy Systems
  23~(7) (2001) 517--530.

\bibitem{Borsche15}
T.~Borsche, T.~Liu, D.~J. Hill, Effects of rotational inertia on power system
  damping and frequency transients, 54th Annual Conference on Decision and
  Control (CDC) (2015) 5940--5946.

\bibitem{Liao17}
K.~Liao, Z.~He, Y.~Xu, G.~Chen, Z.~Dong, K.~Wong, A sliding mode based damping
  control of {DFIG} for interarea power oscillations, IEEE Trans. Sustainable
  Energy 8~(1) (2017) 258 -- 267.

\bibitem{Yagooti16}
A.~Yaghooti, M.~Buygi, M.~Shanechi, Designing coordinated power system
  stabilizers: A reference model based controller design, IEEE Trans. Power
  Systems 31~(4) (2016) 2914 -- 2924.

\bibitem{Liu16}
Y.~Liu, Q.~H. Wu, X.~X. Zhou, Coordinated switching controllers for transient
  stability of multi-machine power systems, IEEE Trans. Power Systems 31~(5)
  (2016) 3937 -- 3949.

\bibitem{Taranto99}
J.~Taranto, A.~do~Bomfim, D.~Falcao, N.~Martins, Automated design of multiple
  damping controllers using genetic algorithms, Proc. IEEE Power Engineering
  Society. Winter Meeting (1999) 539 -- 544.

\bibitem{Fuchs14}
A.~Fuchs, M.~Imhof, T.~Demiray, M.~Morari, Stabilization of large power systems
  using {VSC}-{HVDC} and model predictive control, IEEE Trans. Power Delivery
  29~(1) (2014) 480 -- 488.

\bibitem{lei2001optimization}
X.~Lei, E.~Lerch, D.~Povh, Optimization and coordination of damping controls
  for improving system dynamic performance, IEEE Trans. Power Systems 16~(3)
  (2001) 473--480.

\bibitem{obaid2017power}
Z.~A. Obaid, L.~Cipcigan, M.~T. Muhssin, Power system oscillations and control:
  Classifications and {PSSs}' design methods: A review, Renewable and
  Sustainable Energy Reviews 79 (2017) 839--849.

\bibitem{marinescu2019residue}
B.~Marinescu, Residue phase optimization for power oscillations damping control
  revisited, Electric Power Systems Research 168 (2019) 200--209.

\bibitem{kammer2017decentralized}
C.~Kammer, A.~Karimi, Decentralized and distributed transient control for
  microgrids, IEEE Trans. Cont. Syst. Tech.~(99) (2017) 1--12.

\bibitem{doyle1989state}
J.~C. Doyle, K.~Glover, P.~P. Khargonekar, B.~A. Francis, State-space solutions
  to standard $\mathcal{H}_2$ and $\mathcal{H}_\infty$ control problems, IEEE
  Trans. Automatic control 34~(8) (1989) 831--847.

\bibitem{gahinet1994linear}
P.~Gahinet, P.~Apkarian, A linear matrix inequality approach to
  $\mathcal{H}_\infty$ control, International journal of robust and nonlinear
  control 4~(4) (1994) 421--448.

\bibitem{Scherer13}
C.~W. Scherer, Structured $\mathcal{H}_\infty$ optimal control for nested
  interconnections: A state-space solution, Systems \& Control Letters (2013)
  1105--1113.

\bibitem{apkarian2018structured}
P.~Apkarian, D.~Noll, Structured $\mathcal{H}_\infty$-control of
  infinite-dimensional systems, International Journal of Robust and Nonlinear
  Control 28~(9) (2018) 3212--3238.

\bibitem{Hassibi1999}
A.~Hassibi, J.~How, S.~Boyd, A path-following method for solving {BMI} problems
  in control, in: American Control Conference (ACC), Vol.~2, IEEE, 1999, pp.
  1385--1389.

\bibitem{ibaraki2001rank}
S.~Ibaraki, M.~Tomizuka, Rank minimization approach for solving {BMI} problems
  with random search, in: American Control Conference (ACC), Vol.~3, IEEE,
  2001, pp. 1870--1875.

\bibitem{dinh2012combining}
Q.~Dinh, S.~Gumussoy, W.~Michiels, M.~Diehl, Combining convex--concave
  decompositions and linearization approaches for solving {BMIs}, with
  application to static output feedback, IEEE Trans. Automatic Control 57~(6)
  (2012) 1377--1390.

\bibitem{Han04}
J.~Han, R.~Skelton, An {LMI} optimization approach for structured linear
  controllers, 42nd IEEE International Conference on Decision and Control~(5)
  (2004) 5143 -- 5148.

\bibitem{Karimi07}
A.~Karimi, H.~Khatibi, R.~Longchamp, Robust control of polytopic systems by
  convex optimization, Automatica 43~(8) (2007) 1395--1402.

\bibitem{schuler2011design}
S.~Schuler, M.~Gruhler, U.~M{\"u}nz, F.~Allg{\"o}wer, Design of structured
  static output feedback controllers, {IFAC} Proceedings Volumes 44~(1) (2011)
  271--276.

\bibitem{befekadu2006robust}
G.~Befekadu, I.~Erlich, Robust decentralized controller design for power
  systems using matrix inequalities approaches, in: Power Eng. Soc. General
  Meet., IEEE, 2006.

\bibitem{Mesanovic18ACC}
A.~Me{\v{s}}anovi{\'c}, D.~Unseld, U.~M{\"u}nz, C.~Ebenbauer, R.~Findeisen,
  Parameter tuning and optimal design of decentralized structured controllers
  for power oscillation damping in electrical networks, in: Proc. Amer. Cont.
  Conf. (ACC), IEEE, 2018, pp. 3828--3833.

\bibitem{HIFOO}
S.~Gumussoy, D.~Henrion, M.~Millstone, M.~L. Overton, Multiobjective robust
  control with {HIFOO} 2.0, Proc. 6th {IFAC} Symposium on Robust Control Design
  42~(6) (2009) 144--149.

\bibitem{apkarian2006nonsmooth}
P.~Apkarian, D.~Noll, Nonsmooth {H}-infinity synthesis, IEEE Trans. Automatic
  Control 51~(1) (2006) 71 -- 86.

\bibitem{kanev2004robust}
S.~Kanev, C.~Scherer, M.~Verhaegen, B.~De~Schutter, Robust output-feedback
  controller design via local {BMI} optimization, Automatica 40~(7) (2004)
  1115--1127.

\bibitem{boyd2016mimo}
S.~Boyd, M.~Hast, K.~{\AA}str{\"o}m, {MIMO} {PID} tuning via iterated {LMI}
  restriction, International Journal of Robust and Nonlinear Control 26~(8)
  (2016) 1718--1731.

\bibitem{Mesanovic17ISGT}
A.~Me{\v s}anovi{\'c}, U.~M{\"u}nz, R.~Findeisen, Coordinated tuning of
  synchronous generator controllers for power oscillation damping, in:
  Innovative Smart Grid Technologies Conference Europe, IEEE, 2017.

\bibitem{mesanovic18ACDC}
A.~Me{\v s}anovi{\' c}, U.~M{\" u}nz, R.~Findeisen, Coordinated tuning of
  controller parameters in {AC/DC} grids for power oscillation damping, in:
  {IEEE/PES} Transmission and Distribution Conference and Exposition, 2018.

\bibitem{lunze2013regelungstechnik}
J.~Lunze, Regelungstechnik 2: Mehrgr{\"o}{\ss}ensysteme Digitale Regelung,
  Springer-Verlag, 2013.

\bibitem{poolla2019placement}
B.~Poolla, D.~Gross, F.~D{\"o}rfler, Placement and implementation of
  grid-forming and grid-following virtual inertia and fast frequency response,
  IEEE Trans. Pow. Sys. (2019).

\bibitem{pddotnuschel2018frequency}
S.~P{\"u}schel-L{\o}vengreen, P.~Mancarella, Frequency response constrained
  economic dispatch with consideration of generation contingency size, in:
  Power Systems Computation Conference (PSCC), IEEE, 2018.

\bibitem{schiffer2016survey}
J.~Schiffer, D.~Zonetti, R.~Ortega, A.~M. Stankovi{\'c}, T.~Sezi, J.~Raisch, A
  survey on modeling of microgrids-from fundamental physics to phasors and
  voltage sources, Automatica 74 (2016) 135--150.

\bibitem{MathworksExciter}
MathWorks, Excitation system,
  \url{https://de.mathworks.com/help/physmod/sps/powersys/ref/excitationsystem.html}
  (Jan 2020).

\bibitem{moeini2015open}
A.~Moeini, I.~Kamwa, P.~Brunelle, G.~Sybille, Open data {IEEE} test systems
  implemented in simpowersystems for education and research in power grid
  dynamics and control, in: Power Engineering Conference (UPEC), 2015 50th
  International Universities, IEEE, 2015, pp. 1--6.

\bibitem{TGOVMathworks}
MathWorks, Steam turbine and governor,
  \url{https://www.mathworks.com/help/physmod/sps/powersys/ref/steamturbineandgovernor.html}
  (Jan 2020).

\bibitem{IEEEExciters06}
{IEEE}, {IEEE} recommended practice for excitation system models for power
  system stability studies, {IEEE} Std 421.5-2005 (Revision of {IEEE} Std
  421.5-1992) (2006) 1--93\href {https://doi.org/10.1109/IEEESTD.2006.99499}
  {\path{doi:10.1109/IEEESTD.2006.99499}}.

\bibitem{rahmoun2017mathematical}
A.~Rahmoun, A.~Armstorfer, H.~Biechi, A.~Rosin, Mathematical modeling of a
  battery energy storage system in grid forming mode, in: Power and Electrical
  Engineering of Riga Technical University (RTUCON), 58th Intl. Sci. Conf. on,
  IEEE, 2017, pp. 1--6.

\bibitem{sinamics}
Sinamics inverter,
  \url{https://www.siemens.com/global/en/home/products/drives/sinamics.html}
  (Jan 2020).

\bibitem{boyd1985subharmonic}
S.~Boyd, C.~Desoer, Subharmonic functions and performance bounds on linear
  time-invariant feedback systems, {IMA} Journal of Mathematical control and
  Information 2~(2) (1985) 153--170.

\bibitem{boyd1991linear}
S.~Boyd, C.~Barratt, Linear controller design: limits of performance, Tech.
  rep., Stanford University Stanford United States (1991).

\bibitem{zhou1998essentials}
K.~Zhou, J.~Doyle, Essentials of robust control, Prentice hall Upper Saddle
  River, NJ, 1998.

\bibitem{uherka1977continuous}
D.~Uherka, A.~Sergott, On the continuous dependence of the roots of a
  polynomial on its coefficients, The American mathematical monthly 84~(5)
  (1977) 368--370.

\bibitem{de1989analytic}
B.~De~Moor, S.~Boyd, Analytic properties of singular values and vectors, KTH
  Leuven, Belgium Tech. Rep 28 (1989) 1989.

\bibitem{nocedal2006numerical}
J.~Nocedal, S.~Wright, Numerical optimization, Springer, 2006.

\bibitem{Yalmip}
J.~Lofberg, {YALMIP} : a toolbox for modeling and optimization in {MATLAB},
  IEEE Int. Conf. on Robotics and Automation (2004) 284 -- 289.

\bibitem{Sedumi}
J.~F. Sturm, \href{https://sedumi.ie.lehigh.edu/?page\_id=58}{Using {SeDuMi}
  1.05, a {Matlab} toolbox for optimization over symmetric cones} (Jan 2020).
\newline\urlprefix\url{https://sedumi.ie.lehigh.edu/?page\_id=58}

\bibitem{SimPowSys}
Mathworks, Simscape power systems,
  \url{https://www.mathworks.com/products/simpower.html} (Jan 2020).

\bibitem{mesanovic2018ISGT}
A.~Me{\v{s}}anovi{\'c}, U.~M{\"u}nz, J.~Bamberger, R.~Findeisen, Controller
  tuning for the improvement of dynamic security in power systems, in: IEEE PES
  Innovative Smart Grid Technologies Conf. Europe, IEEE, 2018.

\bibitem{dynagrid}
DynaGridCenter, Dynagridcenter, \url{http://forschung-stromnetze.info/projekte/
  dynamische-stromnetze-sicher-beherrschen/} (Jan 2020).

\bibitem{iren2}
{IREN2} - future oriented electricity grids for integration of renewable energy
  systems, \url{http://www.iren2.de/en/} (Jan 2020).

\bibitem{gross2019effect}
D.~Gro{\ss}, M.~Colombino, J.~Brouillon, F.~D{\"o}rfler, The effect of
  transmission-line dynamics on grid-forming dispatchable virtual oscillator
  control, IEEE Trans. Cont. of Network Syst. (2019).

\bibitem{chen2019parameter}
J.~Chen, T.~O'Donnell, Parameter constraints for virtual synchronous generator
  considering stability, IEEE Trans. Power Syst. 34~(3) (2019) 2479--2481.

\bibitem{mesanovic2019comparison}
A.~Me{\v{s}}anovi{\'c}, X.~Wu, S.~Schuler, U.~M{\"u}nz, F.~D{\"o}rfler,
  R.~Findeisen, Optimal design of distributed controllers for large-scale
  cyber-physical systems, in: Design Automation of Cyber-Physical Systems,
  Springer, 2019, pp. 181--210.

\bibitem{systune}
Mathworks, systune,
  \url{https://www.mathworks.com/help/control/ref/systune.html} (Jan 2020).

\end{thebibliography}

\appendix


\section{Application of Theorem~\ref{thm.Stability}}
\label{App.ExampleSystem}
To underline the claim of Theorem~\ref{thm.Stability}, we consider the small system
\begin{align}
G'(s) = \begin{pmatrix}
\frac{s+2}{(s+1)(s+3)} & \frac{s-3}{s^2 + 3s + 3} \\
\frac{s^2+4s + 10}{(s+3) (s^2 + s + 1)} & \frac{s+4}{(s+1)(s+2)}.	
\end{pmatrix}
\end{align}
This system has the pole set $\tilde \cS = \{-1, -2, -3, -1.5 \pm j 0.87, -0.5 \pm j 0.87\}$, where the poles $s = - 1$, and $s= - 2$ have a multiplicity of 2.
Figure~\ref{fig.examplePoles} shows the largest singular value of $G'(s )$.
\begin{figure}[tb]
	\centering
	\includegraphics[width=1\columnwidth]{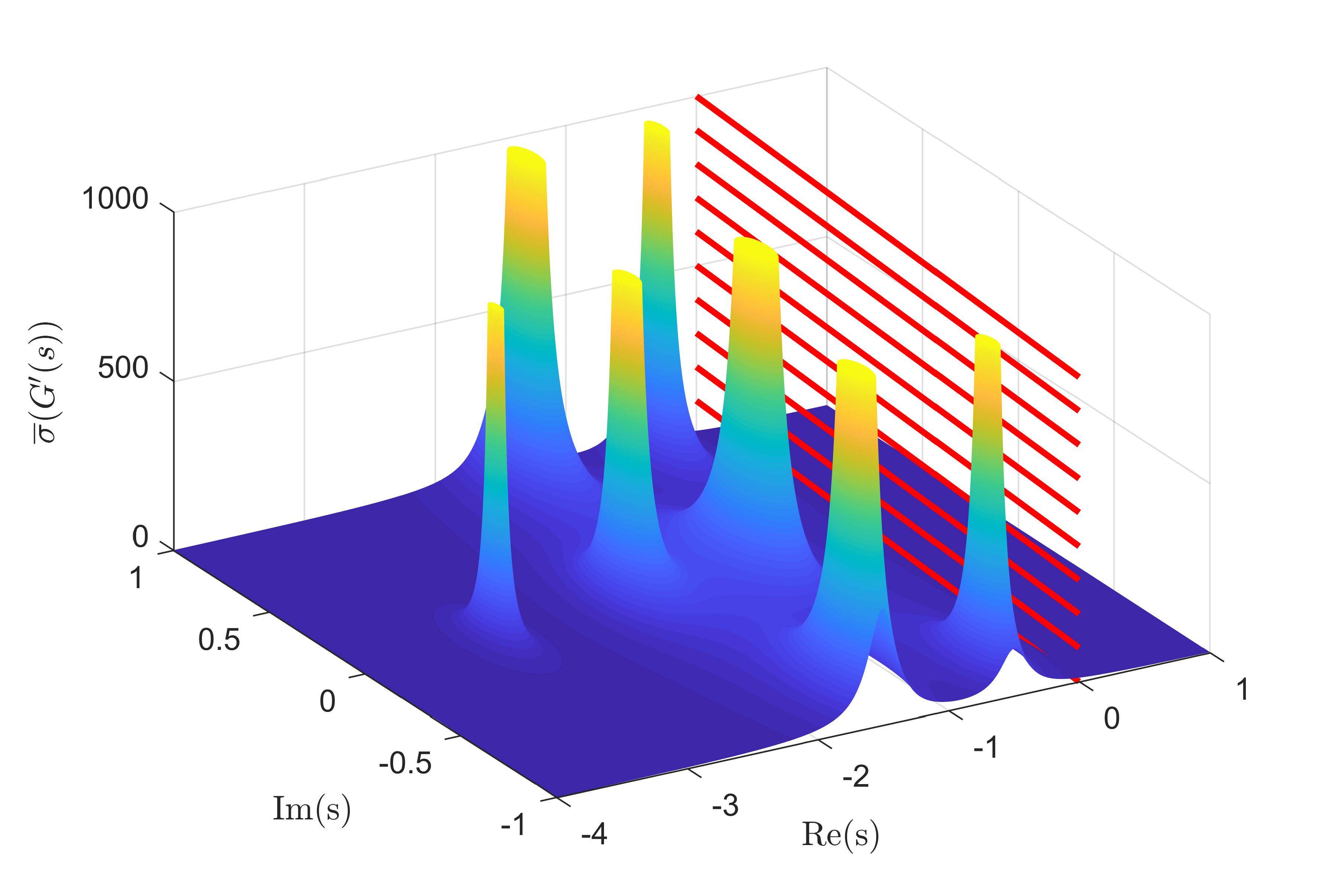}
	\caption{Graphical representation of $\bigSigma(G'(s))$; $\bigSigma(G'(s))$ approaches infinity in the surrounding of any $s_{pij}\in \tilde \cS$. Red lines mark where the system $\Hinf$ norm is minimized.}
	\label{fig.examplePoles}
\end{figure}
\begin{table}[tb]
	\centering
	\caption{{The initial step size $\vDK$ for the controller parameters of each power plant in the IEEE 39 bus system, shown in Figs.~\ref{fig.Exciter}-~\ref{fig.TGOV}.}}
	\label{tab.DKIEEE39}
	\begin{tabular}{ccccccccccc}
		\toprule       
		& $R_{p,i}$ & $K_{A,i}$ & $T_{fd,i}$ & $K_{fd,i}$ & $K_{S,i}$ & $T_{w,i}$ & $T_{1,i}$ & $T_{2,i}$ & $T_{3,i}$ & $T_{4,i}$ \\ \midrule
		$\vDK$&0.05    & 60      & 500       & 1       &      6     &      50  &    1      &     50      &    1       &     50      \\ \bottomrule
	\end{tabular}
\end{table}
It confirms that the system singular values approach infinity as s approaches one of the system poles, see Lemma~\ref{lem.PoleLimitSV}. 
We minimize the $\Hinf$ norm of the system by minimizing the largest singular value of $G'$ on the imaginary axis, i.e. $\bigSigma(G'(j \omega))$.
The plane with $\myRe(s)=0$, along which $\bigSigma(G'(j \omega))$ is minimized, is represented with red lines in Fig.~\ref{fig.examplePoles}.
If the poles approach the imaginary axis, $\max_{\omega \in \R} \| \bigSigma (G'(j\omega)) \|_\infty$ rises to large values. Since the $\Hinf$ norm is minimized in every optimization step, the minimization of the $\Hinf$ norm will never lead to the system poles reaching, and crossing, the imaginary axis.


\section{Controller models used for the IEEE 39 bus 10 power plant model}
\label{App.IEEE39Models}

Figures~\ref{fig.Exciter},~\ref{fig.PSS}, and~\ref{fig.TGOV} show the power plant controller models used for modeling of the IEEE 39 bus grid in Subsection~\ref{subsec.IEEE39}. All models are a part of the system proposed in~\cite{moeini2015open}. We optimize the gain $K_{A,i}$ of the AVR$_i$, shown red in Fig.~\ref{fig.Exciter}. We also optimize all parameters of PSS$_i$, except the physically-determined sensor time constant, marked red in Fig.~\ref{fig.PSS}. The governor and turbine model, shown in Fig.~\ref{fig.TGOV}, has one optimization parameter, marked in red. It is the proportional gain of the governor.
All presented models are standard IEEE models.  {The initial maximal allowed step size for all controllers is shown in Table~\ref{tab.DKIEEE39}.}

\begin{figure}[tb]
	\centering
	\includegraphics[width= 0.7\columnwidth]{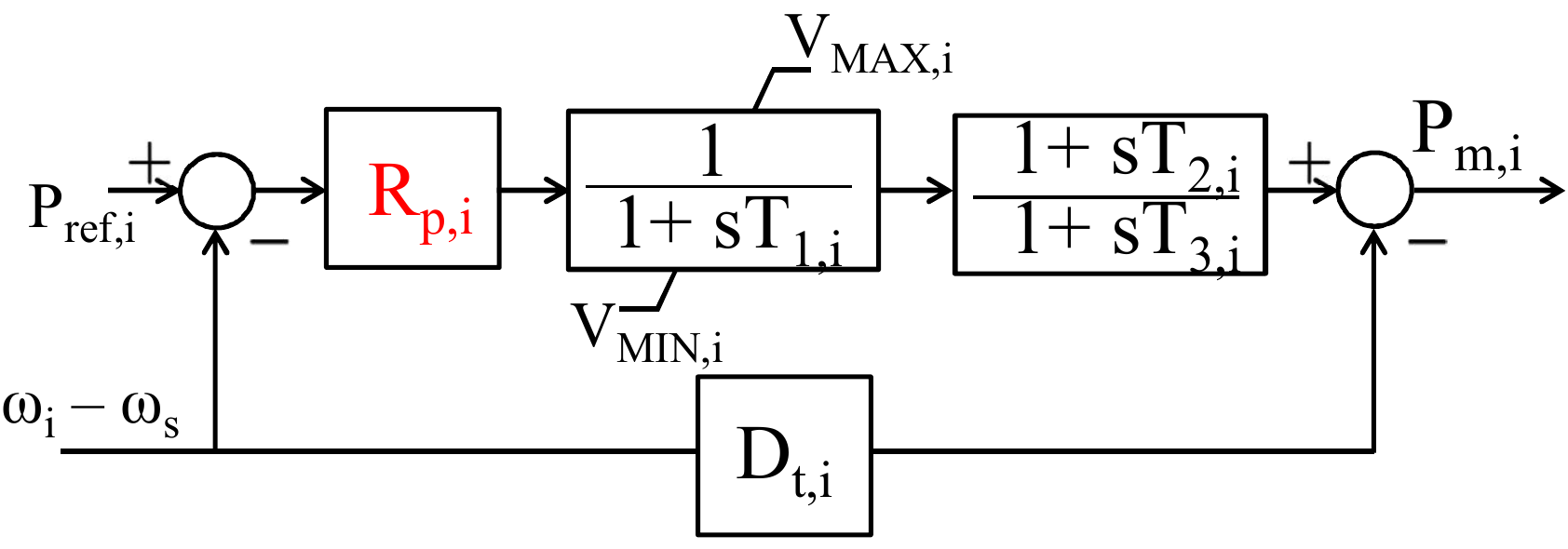}
	\caption{The TGOV1 turbine and governor model used for the power system model in Subsection~\ref{subsec.DynaGrid}. The frequency droop gain of the governor $R_{p,i}$ is an optimization variable.}
	\label{fig.TGOV1Model}
\end{figure}

\begin{figure}[tb]
	\centering
	\includegraphics[width=0.8\columnwidth]{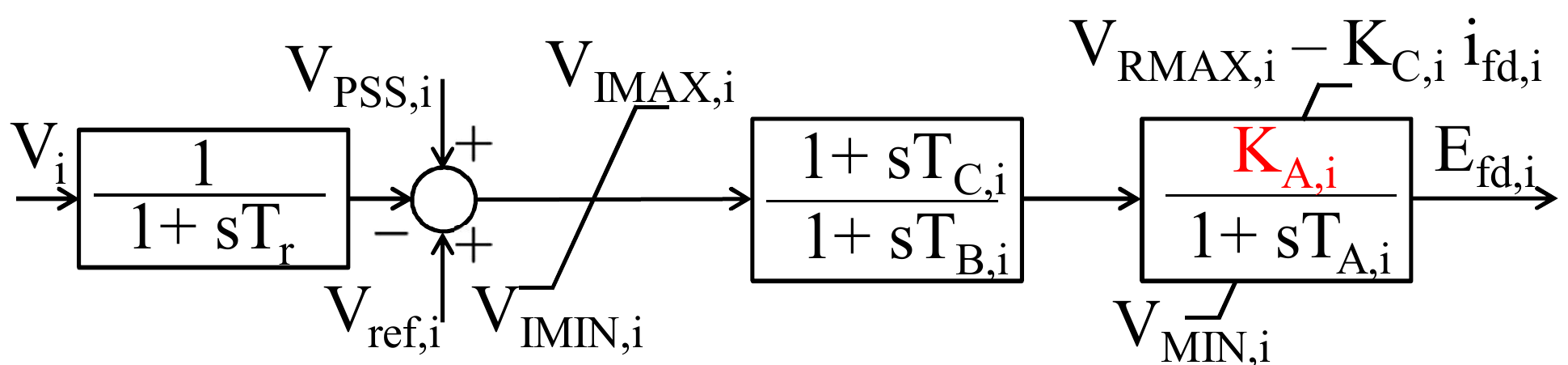}
	\caption{ The standard EXAC4 model of the AVR$_i$, where $T_{r,i}$ is the transducer time constant, $T_{C,i}$ and $T_{B,i}$ are dynamic gain reduction time constants, $K_{A,i}$ is the AVR gain, and $T_{A,i}$ is the AVR lag time constant. We assume that $K_{A,i}$, marked red, is tunable.}
	\label{fig.EXAC4}
\end{figure}

\begin{figure}[tb]
	\centering
	\includegraphics[width=1.0\columnwidth]{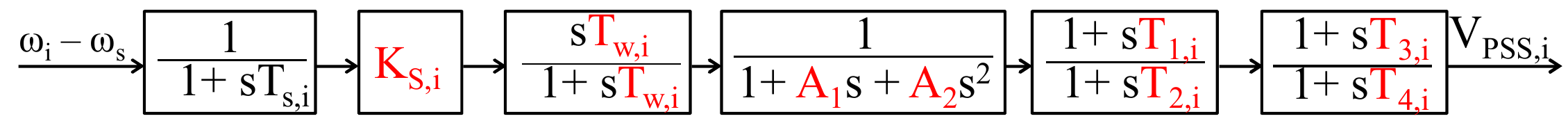}
	\caption{The standard IEEE PSS 1A model, where $K_{S,i}$ is the PSS gain, $T_{w,i}$ is the washout time constant, $T_{1,i}$-$T_{4,i}$ are the lead-lag filters time constants, $T_{s,i}$ is the sensor time constant, and $A_1$ and $A_2$ are notch filter parameters. All of the PSS parameters are tunable, except the sensor time constant.}
	\label{fig.IEEEPSS1A}
\end{figure}

\section{Controller models used for the 53 generator power system model}
\label{App.DynaGridModels}

The reduced European grid defined in Subsection~\ref{subsec.DynaGrid} uses controllers shown in Figs.~\ref{fig.TGOV1Model},~\ref{fig.EXAC4}, and~\ref{fig.IEEEPSS1A}. Similar to the IEEE 39 bus controller models, the gains of TGOV$_i$ and AVR$_i$ are tuned, as well as all parameters of PSS$_i$. For this power system, the standard model TGOV1 is used for TGOV$_i$, the EXAC4 model is used for AVR$_i$, and the IEEE PSS 1A model is used for PSS$_i$. All presented controller models are standard IEEE models.

\end{document}